\documentclass[manuscript,screen]{acmart}

\usepackage{type1cm}
\usepackage{times}

\usepackage{helvet}
\usepackage{courier}
\usepackage{url}
\usepackage{bm}
\usepackage{graphicx}
\usepackage{tcolorbox}
\usepackage{enumerate}
\usepackage{pgfplots}
\usepackage{algorithmic}
\usepackage{algorithm}

\usepackage{eqparbox}
\usepackage{xcolor}

\usepackage{eucal}
\usepackage{mathtools}
\usepackage{float}
\usepackage{color}
\usepackage{bbold}
\usepackage{xspace}
\usepackage{lettrine}
\usepackage{subfig}
\usepackage{tikz}
\usetikzlibrary{backgrounds,automata,shapes.geometric, arrows, positioning}
\usepackage{latexsym}
\usepackage{fancybox}

\usepackage{thm-restate}

\newtheorem{lemma}{Lemma}

\newtheorem{corollary}{Corollary}
\newtheorem{observation}{Observation}
\theoremstyle{remark}
\newtheorem{remark}{Remark}
\usepackage{hyperref}
\usepackage{cleveref}

\usepackage{setspace}
\usepackage{ifthen}
\usepackage{caption}

\newcommand{\Hg}{\mathcal{H}}
\newcommand{\So}{\tilde{S}}
\newcommand{\G}{\tilde{G}}
\newcommand{\hv}{\tilde{v}}
\newcommand{\hi}{j}
\newcommand{\U}{\mathcal{P}}
\newcommand{\X}{\mathbf{X}}
\newcommand{\N}{\mathbf{N}}
\renewcommand{\vv}{\mathbf{v}}
\newcommand{\hvv}{\tilde{\mathbf{v}}}
\newcommand{\prpt}{\Pi}

\newcommand{\FixPropMultTwo}{FixProp12}
\newcommand{\MinEnviedSet}{MinimalEnviedSet}

\AtBeginDocument{%
  }

\setcopyright{acmlicensed}
\copyrightyear{2018}
\acmYear{2018}
\acmDOI{XXXXXXX.XXXXXXX}
\acmConference[Conference acronym 'XX]{Make sure to enter the correct
  conference title from your rights confirmation email}{June 03--05,
  2018}{Woodstock, NY}
\acmISBN{978-1-4503-XXXX-X/2018/06}

\begin{document}

\title{Almost EFX in Hypergraphs}

\author{Ioannis Kakatelis}
\affiliation{%
  \institution{Czech Technical University in Prague}
  \city{}
  \country{Czech Republic}}
\email{ioannis.kakatelis@fit.cvut.cz}

\author{Thanasis Lianeas}
\affiliation{%
  \institution{University of West Attica}
  \country{Greece}}
\email{lianeas@corelab.ntua.gr}

\author{Alkmini Sgouritsa}
\affiliation{%
  \institution{Athens University of Economics and Business, and Archimedes/Athena RC }
  \country{Greece}
}
\email{alkmini@aueb.gr}

\author{Minas Marios Sotiriou}
\affiliation{%
  \institution{Athens University of Economics and Business, and Archimedes/Athena RC }
  \country{Greece}
}
\email{minas\_marios@outlook.com}
\renewcommand{\shortauthors}{Kakatelis et al.}

\begin{abstract}
We study the existence of envy-free-up-to-any-good (EFX) allocations of {\em indivisible} goods among agents with heterogeneous monotone valuations.
Christodoulou et al. (2023) introduced the (multi-hyper)graph setting, where agents and goods are represented by vertices and edges of a graph respectively, and only the endpoints of an edge may have non-zero marginal value for it. Our work simplifies and extends previous results of \cite{kaviani2024envyfreeallocationindivisiblegoods} in this domain. First, we provide a simpler construction of EF2X allocation for \textit{general monotone} valuations   in \textit{hypergraphs} with girth at least 3. We extend our ideas when the multiplicity of each edge is 2 and show that an EF3X allocation always exists for additive valuations. Both results can be constructed in polynomial time. Regarding EFX approximations, we provide a simpler construction for $\frac{\sqrt{2}}{2}$-EFX allocations in hypergraphs of girth at least 3 under subadditive valuations. We push the state-of-the-art by establishing the existence of $\frac{2}{3}$-EFX allocations for additive valuations when the edge multiplicity is 2. Both of the latter results can be constructed in pseudo-polynomial time. By addressing these multi-hypergraph settings, our work contributes to the ongoing effort to resolve the existence of EFX in increasingly general and applicable domains.

\end{abstract}



\maketitle

\section{Introduction}
The problem of fair division involves allocating a set of indivisible goods among multiple agents in a manner that ensures fairness. While the roots of fair resource allocation trace back to antiquity, the field formally emerged in the late 1940s with the work of \citet{h__steihaus_1948} and later the works of \cite{gamow1958puzzle, foley1966resource, VARIAN197463}. Today, fair division represents a rich interdisciplinary field intersecting economics, social choice theory, mathematics, and computer
science \cite{ dblp:journals/ai/amanatidisabflmvw23, brams1996fair, brandt2016handbook, robertson1998cake}. In the last decade, research has focused on minimizing envy among agents in diverse applications such as auctions, business partnerships, and course scheduling \cite{courses}. However, eliminating envy is often impossible with indivisible items; for instance, if two agents both desire a single indivisible good, any allocation necessarily creates envy. This reality has necessitated the study of envy-freeness relaxations.
The first major relaxation is Envy-freeness up to one good (EF1) introduced by \citet{budish}, which requires that envy can be eliminated by removing a prespecified single good from the envied agent's bundle. \citet{LiptonEtAl} proved that EF1 allocations are guaranteed to exist even for general monotone valuations and can be computed efficiently. A more rigorous and ``enigmatic'' relaxation is Envy-freeness up to any good (EFX) as mentioned in \cite{procaccia} and was introduced by \citet{EFXCara} and \citet{GMT14}. 
In an EFX allocation an agent should not envy another agent after the hypothetical removal of any single good from the envied agent's bundle. 

Whether EFX allocations always exist remained one of the most significant open problems in fair division for several years. Very recently, \citet{akrami2026counterexample} answered negatively the question of EFX existence by giving a counterexample in the case of $n \geq 3$ agents and $m \geq n+5$ goods. In the case of $n = 3$ and $m = 8$, where agents have general monotone valuations, this result was extended to show non-existence even for submodular valuations \cite{mackenzie2026counterexamplesefxsubmodularsubadditive}. The question still remains open for more specialized valuation functions, such as additive, and more restrictive but realistic settings. Known existence results are limited to specific cases: when there are two agents with general monotone valuations and for any number of agents, provided that all agents have identical valuations \cite{PlautRough} and also for three agents with additive valuations \cite{CGM24}. The case of three agents was extended to nice-cancelable valuations by \citet{almostfullefxforfouragents} and generalized further by \citet{akrami2025efx} to the  case where two agents have arbitrary monotone valuations and  at least one agent has MMS-feasible valuation. 
For $n$ agents, known results involve cases where agents have additive bi-valued valuations \cite{twovaluedinstanses,byrka2026probing}, all agents have one of two general valuations or when the number of items is at most $n+3$ \cite{mahara2021extensionadditivevaluationsgeneral}, when the valuations are submodular with binary marginals \cite{babaioff2020fairtruthfulmechanismsdichotomous}, and when there are at most three distinct additive valuations \cite{efxexistsfor3typesofadditiveagents}. 

The attention of the research community has also turned into graphical valuations, where the relationship between agents and goods is captured by a graph $G=(V,E)$. In this setting, agents are represented by vertices and goods by edges; an edge/good may have positive value only for its endpoints
\cite{EFXsimplegraphs}. \citet{EFXsimplegraphs} proved that EFX allocations always exist on simple graphs and they noted that EFX orientations, i.e., allocations where each edge is given to one of its endpoints may not always exist and are NP-complete to decide their existence. 
Our work focuses on multigraph and hypergraph settings, where a single good may be relevant to more than two agents and multiple goods may connect the same set of agents. These structures are more complex but offer deeper insights into real-world applications like territorial borders or shared research workspace.

\subsection{Our Contribution}
We focus on hypergraphs of girth at least 3 (meaning that any two agents share at most one good/edge) and relax the EFX requirement in two directions, searching for EF2X allocations, i.e., allocations under which removing any two goods from any envied agent's bundle eliminates the envy towards the agent, and approximate EFX allocations. In both directions we recover best-known results via simpler constructions. We further extend these findings to  multi-hypergraph instances of girth at least 3 and multiplicity 2,  i.e., instances where any subset of agents may share more than one but up to two common goods. Notably the only prior work that considers multiplicity greater than 1 is by  \citet{AAAILianeasSS26}, who showed EFX existence for girth at least 4. In contrast, our focus on girth at least 3 introduces significantly more complex agent connectivity, resulting in a more challenging setting.

In Section \ref{sec:EFkX} we aim for EF2X allocations and provide a simple polynomial time algorithm that  returns an EF2X allocation for hypergraphs of girth at least 3 and general monotone valuations (Theorem \ref{thm:girth3}), thus reproving a result of   \citet[Theorem 5.7]{kaviani2024envyfreeallocationindivisiblegoods} with a much simpler and polynomial time construction. The algorithm, in a first step, by employing ideas of \citet{AAAILianeasSS26} and prioritizing the agents, creates first a partial EFX allocation with specific properties, and in a second step it assigns, in one pass,  the unallocated goods to reach  a complete EF2X allocation. Next, moving one step forward, we focus on multi-hypergraphs of girth at least 3 and multiplicity 2. For such instances, we provide a polynomial time algorithm that returns an EF3X allocation for additive valuations (Theorem \ref{thm:girth3EF3X}), i.e., any agent does not envy any other bundle after removing 3 goods from it.

In Section \ref{sec:apxEFX} we turn into finding approximate EFX allocations. For simple hypergraphs of girth at least 3 we  extend ideas used by Amanatidis et al. \cite{amanatidis2024pushingfrontierapproximateefx} and show that a  $\frac{\sqrt{2}}{2}$-EFX allocation  exists and can be computed in pseudo-polynomial time for subadditive valuations in hypergraphs with girth at least 3 (Theorem \ref{thmhyper1}).
Our work matches the current best known approximation result for hypergraphs \cite{kaviani2024envyfreeallocationindivisiblegoods}, but the construction is 
simpler, requiring a more straightforward envy graph construction. 
Furthermore, we extend the frontier in multi-hypergraphs, demonstrating that a $\frac{2}{3}$-EFX allocation exists and can be computed in pseudo-polynomial time for additive valuations in multi-hypergraphs of girth at least 3 and multiplicity 2 (Theorem \ref{thm:hyper2}). The proposed algorithm for this last result needs to be more sophisticated for its final allocation compared to the previous algorithms. It may have  to iterate many times, improving the partial allocation, before getting a partial allocation for which the unallocated goods can be assigned in one last pass, which is what the rest of the algorithms do.

\subsection{Related Work}

\noindent{\bf Approximate EFX}. A significant line of research has considered approximations of EFX, denoted by $\alpha-$EFX, where the envy that remains even after the removal of any good from another envied agent's bundle, is bounded by a multiplicative factor of $a \in(0,1]$. The approximation of EFX was proposed by Plaut and Roughgarden and they showed in \cite{PlautRough} that a $\frac{1}{2}-$EFX allocation always exists when agents have subadditive valuations; it was later shown by \citet{DBLP:conf/atal/Chan0LW19} that this allocation can be derived in polynomial time. The best known result in EFX approximation guarantee under the assumption that the agents have additive valuations, without any restriction on the number of agents or their values is $\alpha = \frac{1}{\phi} \approx 0.618 $ by \citet{Amanatidis_Markakis_Ntokos_2020}. In more restricted valuation functions, \citet{DBLP:conf/atal/MarkakisS23} showed that a $\frac{2}{3}-$EFX allocation exists when the agents agree on what are the top n items (but not necessarily on their exact ranking), with n being the number of agents. Following this direction \citet{amanatidis2024pushingfrontierapproximateefx} showed that a $\frac{2}{3}-$EFX allocation exists when each agent values each good by one of three fixed values, or when there are at most 7 agents, which was recently improved to 8 agents by \citet{filosratsikas2026approximateenvyfreeallocationsk}. Moreover, Prakash \citet{hv2025almost} showed the existence of $\frac{2}{3}-$EFX allocation when there are at most four distinct valuation functions. On the negative side, very recently, \citet{mackenzie2026counterexamplesefxsubmodularsubadditive} showed that $\alpha-$EFX does not exist for three agents with monotone subadditive valuations and eight goods for $\alpha \approx 0.89$.

\noindent{\bf EFX with Charity}. An important relaxation of EFX that was introduced by \citet{CGH19} is called EFX with charity. In this scenario we are seeking for a (partial) EFX allocation, where not all the items need to be allocated (some of them remain unallocated), or in other words, we donate some goods to charity and achieve EFX with the rest. \citet{CGH19} proved the existence of a partial allocation that is EFX and its Nash Social Welfare is half of the maximum possible. Later, \citet{CKMS21} proved that EFX allocation exist for $n$ agents with arbitrary valuations if up to $n-1$ goods can be donated and no agent envies the donated bundle. \citet{almostfullefxforfouragents} showed that an EFX allocation exists for $n$ agents with nice cancelable valuations if up to $n-2$ goods can be donated. They also showed that an EFX allocation exists for four agents with at most one donated good. \citet{mahara2021extensionadditivevaluationsgeneral} extended the EFX allocation existence to general monotone valuations leaving $n-2$ goods unallocated. For hypergraphs with girth at least 3, the size of charity was reduced to $\big\lfloor \frac{n}{2} \big\rfloor - 1$ for general monotone valuations \cite{kaviani2024envyfreeallocationindivisiblegoods}. Regarding now the number of donated goods, some studies have introduced a correspondence between approximate EFX allocations with a sublinear number of discarded goods with a  combinatorial problem called  the Rainbow Cycle problem, see (i.e \cite{chaudhury2024improving,jahan2022rainbow,berendsohn2022fixed,akrami2025efx,hv2025almost}). 

\noindent{\bf EF2X}. Another relaxation of EFX that has been studied is the notion of EF2X. It is a specific case of EFkX notion which is a general relaxation of EFX, introduced by \citet{akrami2022ef2x}, that allows an amount of envy up to the value of $k$ least valuable goods of a bundle. In a recent work by \citet{Ashuri25EF2X4agents} it has been shown that EF2X exists for four agents with cancelable valuations. In \cite{kaviani2024envyfreeallocationindivisiblegoods} they showed that in hypergraphs with girth at least 3, there always exists an EF2X allocation. Very recently, \citet{filosratsikas2026approximateenvyfreeallocationsk} showed that for any number of agents with additive valuations a $\frac{3}{4}-$EF2X allocation always exists and that can be generalized for every $k\geq 2$ to $\frac{k+1}{k+2}-$EFkX approximation.

\noindent{\bf EFX Allocations in Multigraphs and Hypergraphs}.We focus now on the line of work related to multigraphs and hypergraphs.
\citet{EFXsimplegraphs} introduced the graph setting and showed the existence of EFX allocation when the graph is simple and the edges are goods. This result was extended to the case of mixed manna, where items may be both goods and chores \cite{mixedmanna}. Existence of EFX allocations has been shown in multigraphs where agents have restricted additive valuations \cite{kaviani2024envyfreeallocationindivisiblegoods}, in bipartite multigraphs where agents have additive \cite{afshinmehr2024efxallocationsorientationsbipartite} and cancelable \cite{bhaskar2024efxallocationsmultigraphclasses} valuations, in multigraphs with girth at least 6 for agents with general monotone valuations \cite{OntheExistenceofEFXAllocationsinMultigraphs}. \citet{ZengStructureOrientationsGraphsAAMAS} established a connection between the EFX orientation existence and the chromatic number of the graph, more precisely they showed that EFX orientation exists in $t$-colored multigraphs with girth at least $2t-1$ for agents with cancelable valuations \cite{bhaskar2024efxallocationsmultigraphclasses}, in multigraphs where agents are neighbors to at most roughly one quarter of the other agents and have general monotone valuations \cite{OntheExistenceofEFXAllocationsinMultigraphs},  and in multi-trees for agents with general monotone valuations \cite{bhaskar2024efxallocationsmultigraphclasses}. Subsequently, \citet{afshinmehr2025efxallocationsexisttrianglefree} showed that triangle-free multigraphs (girth at least 4) always admit EFX allocations for general monotone valuations. Another, very recent work showed the existence of EFX allocations in hypergraph and multi-hypergraph setting with girth at least 4 and a further assumption on a single's vertex multiplicity \cite{AAAILianeasSS26}. 

\noindent{\bf Approximate EFX in Multigraphs and Hypergraphs}. Approximate EFX-allocations have also been studied in multigraphs and hypergraphs. \citet{amanatidis2024pushingfrontierapproximateefx} showed that $\frac{2}{3}-$EFX allocations always exist in multigraphs when agents have additive valuations, which was recently improved to $\frac{\sqrt{2}}{2}-$EFX \cite{kaviani2025improvedapproximateefxguarantees}. \citet{kaviani2024envyfreeallocationindivisiblegoods} showed that for hypergraphs with 
 girth at least 3, $\frac{\sqrt{2}}{2}$-EFX allocations always exist when agents have  subadditive valuations.  

\noindent{\bf EFX Orientations in Graphs and Multigraphs}. In the graph and multigraph setting the existence of EFX orientations, i.e., allocations where edges/goods may only be allocated to one of the endpoints, has received much of attention. \citet{EFXsimplegraphs} showed that EFX orientations may not exist by giving a counterexample in a $K_4$ graph, and they showed that even deciding if an EFX orientation exists  is NP-complete. This result holds even if the vertex cover of the graph has size  8, or in multigraphs with only 10 vertices \cite{deligkas2024ef1efxorientations}, which was later improved to a vertex cover of size 4 or multigraphs with as few as 4 vertices \cite{kanellopoulos2025ef}. \citet{ZengStructureOrientationsGraphsAAMAS} showed that EFX orientations may not 
exist in graphs with chromatic number greater than 3, and that EFX orientations always exist when the chromatic number is at most 2. \citet{filosratsikas2026approximateenvyfreeallocationsk} showed that even EFkX orientations on graphs do not always exists and deciding whether it admits an EFkX orientation is NP-complete. The complexity of orientations has been also explored in multigraphs \cite{hsu2024efxorientationsmultigraphs,afshinmehr2024efxallocationsorientationsbipartite} and recently through the lens of parameterized complexity \cite{blavzej2025tractable,kanellopoulos2025ef}. 

Another fairness notion has also been explored in multigraphs known as maxmin share (MMS) \cite{DBLP:journals/corr/abs-2506-21493, christodoulou2025exactapproximatemaximinshare}.

\section{Preliminaries}
We consider an instance of discrete fair division to be a triple $(N,M,\vv)$, where $N = [n] = \{1,\ldots,n\}$ is the set of $n$  agents,  $M = [m] =  \{1,\ldots,m\}$ is the set of $m$  indivisible goods and $\vv = (v_{1},\ldots,v_{n})$ is a valuation profile represented by a tuple of valuation functions, where $v_{i}:2^{M}\rightarrow \mathbb{R}_{\geq 0}$ and for agent $i \in N$ and $S\subseteq M$, $v_i(S)$ represents $i$'s value for the subset $S$ of goods. 
Throughout, for $k\in \mathbb{N^{+}}$, we let $\left[k\right]:= \{1,2,\ldots,k\}$.

\subsubsection*{Types of Valuations}
We consider valuation functions (or just valuations) that are {\em monotone}, i.e.,  for any $i\in N$ and $S\subseteq T \subseteq M$, it holds that $v_i(S) \leq v_i(T)$, and {\em normalized}, i.e., $v_i(\emptyset)=0$.  For simplicity, for the valuation of $i$ over some good $g$, we write $v_i(g)$ instead of $v_i(\{g\})$. Also, we may say \textit{value} of (set) $S$ for (agent) $i$ to refer to $v_i(S)$.
Based on additional properties of the valuation functions, we give the following classes of valuation function in increasing order of inclusion.  

\noindent{\bf Additive Valuations}. A valuation function $v_{i}$ is additive if $v_{i}(S) = \sum_{g \in S}v_{i}(g)$ for any $S \subseteq M$. \\ 
\noindent{\bf Subadditive Valuations}. A valuation function is subadditive if $v_{i}(S)+ v_{i}(T) \geq v_{i}(S \cup T)$ for any $S,T \subseteq M$. 
 \\
\noindent{\bf General Monotone Valuations}. When a valuation function has no other restriction but monotonicity and normalization, we call it general monotone.

\subsubsection*{Types of Goods}
A good $g \in M$ is {\em  relevant}  to an agent $i \in N$, if adding $g$ to some subset of goods increases the subset's value for $i$, i.e., there exists an $S\subseteq M$ so that $v_i(S\cup \{g\})>v_i(S)$. For $g$ relevant to $i$, if the valuation functions are subadditive, then $v_i(S)+v_i(g)\geq v_i(S\cup \{g\})>v_i(S)$, and thus,  a good $g$ is relevant to $i$ iff $v_i(g)>0$. A good that is not relevant to $i$ is called {\em irrelevant} to $i$.

\noindent{\bf Twin Good}. Let $g\in M$, we call {\em twin good} of $g$ another good different from $g$ that is relevant to the exact same agents as $g$. We will usually denote it as $g_t$.

\subsubsection*{Types of Allocations/Fairness}
An allocation $\X = (X_{1},\ldots,X_{n})$ is an ordered  tuple of mutually disjoint subsets of $M$ (i.e., for $i \neq j$: $X_{i} \cap X_{j} = \emptyset$). We refer to the $X_i$'s as bundles.
An allocation is {\em complete} if $\cup_{i  = 1}^{n}X_{i} = M$ and is {\em partial} otherwise.
For an allocation $\X$ the {\em pool} $\U(\X)$ is the set containing the goods left unallocated by $\X$, i.e. $\U(\X)=M\setminus \cup_{i  = 1}^{n}X_{i}$. 

\noindent{\bf Social Welfare.} Given an allocation $\X$, the Social Welfare is defined as $\sum_iv_i(X_i)$.

\noindent {\bf Orientation.}
An allocation $\X$ is an {\em orientation} if for any $g \in X_i$,  $g$ is relevant to $i$. We say that we orient a good to an agent if we allocate the good to an agent for which it is relevant.

\noindent{\bf Envy}.
Given an allocation $\X$, we say that an agent $i$ {\em envies} another agent $j$ (or bundle $X_j$), if $v_{i}(X_{i}) <v_{i}(X_{j})$.

\noindent{\bf EFX}.
Given an allocation $\X$, we say that agent $i$ is {\em envy free up to any good}, or simply EFX,  towards an agent $j$ (or bundle $X_j$), if  for all $g \in X_{j}:  v_{i}(X_{i}) \geq v_{i}(X_{j}\setminus \{g\})$. An allocation $\X$ is EFX if every agent is EFX towards any other agent.

\noindent{\bf $\alpha$-EFX}.
Given an allocation $\X$, we say that agent $i$ is  $\alpha$-EFX  towards an agent $j$ (or bundle $X_j$), if  for all $g \in X_{j}:  v_{i}(X_{i}) \geq \alpha\cdot v_{i}(X_{j}\setminus \{g\})$. An allocation $\X$ is $\alpha$-EFX if {\em every} agent is $\alpha$-EFX towards any other agent.

\subsubsection*{Envy Graphs}
\noindent Given a partial allocation $\X$ and the valuation functions of the agents we define the corresponding {\em envy graph}  $G(\X) = (\N,E(\X))$ to be a directed graph where the vertex set is the set of agents $N$ and the edge set contains a directed edge $(i,j)$ if and only if  agent $i$ envies agent $j$, i.e., $E(\X) = \{(i,j):v_{i}(X_{i}) <v_{i}(X_{j})\}$. A vertex $s$ of an envy graph $G(\X)$ is a {\em source} of $G(\X)$ if it has  in-degree 0, i.e., for any $(i,j)\in E(\X): s\neq j$.

\subsubsection*{The Hypergraph Setting}
We assume that the setting is modeled on hypergraphs.
A \emph{hypergraph} is a pair $\Hg=(V,E)$ where $V$ is a set of vertices and $E$ is a set of subsets of $V$, called hyperedges or simply edges. The size of an edge is the number of vertices it contains. For an edge $e\in E$ and a vertex $i\in e$ we say that $e$ is incident to $i$ and $i$ is incident to $e$, or $i$ is an endpoint of $e$. If vertices  $i$ and $j$ belong to some edge $e$ (i.e., $i,j\in e$), we say that $i$ and $j$ share $e$ and we call $i$ and $j$ neighbors.

\noindent {\bf Cycles \cite{berge1973graphs} - Girth.}
    A \emph{cycle} of length k is defined by a sequence $(x_1,e_1,x_2,e_2, \dots ,x_k,e_k,x_{k+1})$ such that:
              (i) $x_1,x_2, \dots, x_k$ are distinct vertices and $x_{k+1} = x_{1}$,
           (ii) $e_1,e_2, \dots, e_k$ are  distinct edges and (iii)
         $x_i,x_{i+1} \in e_i$, $\forall i \in [k]$.
       A hypergraph $\Hg$ has \emph{girth} $k$, if $\Hg$'s shortest cycle has length  $k$; if there is no cycle in $\Hg$ the girth is infinity. 

\noindent {\bf Hypergraph Instance}. An instance $(N,M,\vv)$ of the problem can be equivalently modeled via hypergraphs. In a \emph{hypergraph instance}, a hypergraph $\Hg=(V,E)$ is given together with a set $\vv$ of valuation functions. The vertices of  $\Hg$ are the agents and the edges of $\Hg$ are the goods. Edge/good $e$ {\em may} be relevant to the vertices/agents it contains but  is irrelevant to all other vertices/agents.

\noindent {\bf Multi-hypergraph Setting}. In Section \ref{sec:ef2x_multi} and Section \ref{sec:multi-hyper} we allow a more general hypergraph setting.  In the \emph{multi-hypergraph setting}  
the edge set $E$ is allowed to have
more than one edge on the same set of vertices. Such edges correspond to different goods, with possibly  different impact on the valuation functions of the vertices/agents.
For an edge $e$ that appears $k$ times in $E$ we say that $e$ has multiplicity $k$. If an edge $e$ has multiplicity 2 then we call the two appearances  of $e$ twins, or twin goods.
For a multi-hypergraph $\Hg=(V,E)$ we define the girth of $\Hg$ to be the girth of $\Hg'=(V,E')$, where $E'$ is derived from $E$ if we delete all repetitions of the  edges of $E$.

We provide the following two observations for hypergraphs and multi-hypergraphs of girth at least 3. For hypergraphs, any two agents/vertices may share at most one good/edge, and for multi-hypergraphs, any two vertices/agents may share at most one good/edge and its repetitions. Our results always consider 
instances where the corresponding hypergraph instance has girth at least $3$.

\begin{observation}\label{obs:at-most-1-relevant-common-edge}
    In hypergraphs with girth at least 3 
    any two agents may share at most one good.
\end{observation}
\begin{proof}
    On the contrary, let $i,j$ be two agents/vertices that both belong to goods/edges, say,  $e_1$ and $e_2$. Starting with vertex/agent $i$, following edge/good $e_1$ to reach $j$ and then following edge/good $e_2$ to reach $i$ we get a cycle of length 2, i.e., the cycle $(i,e_1,j,e_2,i)$, contradicting the hypothesis that the girth is at least 3.
\end{proof}

\begin{observation}\label{obs:repetition-common-edge}
    In multi-hypergraphs with girth at least 3 
    any two agents may share at most one good and its repetitions.
\end{observation}

\begin{proof}
     On the contrary, let $i,j$ be two vertices/agents that both belong to  edges/goods, say,  $e_1$ and $e_2$, that one is not a repetition of the other. If we delete all repetitions of the edges/goods then in the simplified graph $\Hg'$, w.l.o.g., $e_1$ and $e_2$ will be present. This is a contradiction due to Observation~\ref{obs:at-most-1-relevant-common-edge}. 
\end{proof}

In multi-hypergraphs with multiplicity 2, we assume w.l.o.g. that all edges/goods have multiplicity exactly 2, i.e., if an edge/good has multiplicity 1, then we create a duplicate ``dummy'' one which has 0 marginal value for all vertices/agents. Any EFX allocation in the new graph remains an EFX allocation in the original graph after the removal of the dummy edges/goods.

\section{Pursuing EF2X Allocations}\label{sec:EFkX}

In this section we relax the notion of EFX to the envy-free allocations up to more than one good. We first consider hypergraphs of girth at least 3, where we show that envy-free allocations up to any two goods always exist even for general monotone valuations and can be computed in polynomial time (Section~\ref{sec:EF2X-simple}). Then we consider  multi-hypergraphs of girth at least 3 and multiplicity at most 2, and show that envy-free allocations up to any three goods always exist even for general monotone valuations and can be computed in polynomial time (Section~\ref{sec:EF2X-multi}). Before we continue with our results, we give some further preliminaries specifically for this section.

\noindent{\bf EFkX}. Given an allocation $\X$, we say that agent $i$ is \textit{envy free up to k goods} $(k \geq 1)$, or simply EFkX,  towards an agent $j$ (or bundle $X_j$), if for all $\{g_1,g_2,\dots g_k\} \subseteq X_{j}:  v_{i}(X_{i}) \geq v_{i}(X_{j}\setminus \{g_1,g_2, \dots g_k\})$. An allocation $\X$ is EFkX if every agent is EFkX towards any other agent.
In this section, we search for EF2X allocations in hypergraph and multi-hypergraph instances where the underlying hypergraph has girth at least 3.  For ease of presentation  we need the following definitions.

\noindent{\bf Pool relevant to $i$}.
    For an agent $i$, and an allocation $\X$, we denote the set of unallocated goods relevant to $i$ as $\U_i(\X)$.

\noindent{\bf Special agent $0$}. In all our algorithms in this section, there is an arbitrarily chosen agent that serves special purposes. We call this agent, {\em agent $0$} or simply $0$.

\noindent{\bf Pool irrelevant to $0$}.
    Given an allocation $\X$, we define $\U^*(\X) = \U(\X) \setminus \U_0(\X)$. Moreover, for an agent $i$, we denote the set of unallocated goods relevant to $i$ that are not relevant to agent $0$ to be $\U^*_i(\X)$. Formally $\U^*_i(\X) = \U_i(\X) \setminus \U_0(\X)$. 

\subsection{Results and roadmap}

Our first main result of the section is the following, in which we reprove the result of   \citet[Theorem 5.7]{kaviani2024envyfreeallocationindivisiblegoods} with a much simpler and polynomial time construction.

\begin{restatable*}{thm}{ThmEFtXsimple}
\label{thm:girth3}
Instances on hypergraphs of girth at least 3 with general monotone valuations always admit an EF2X allocation that can be computed in polynomial time.
\end{restatable*}

The construction is very simple and the proof will come as a result of several but simple lemmas.
We pick an arbitrary agent  and call it agent $0$, with the rest of the agents   renamed arbitrarily with numbers 1 to n-1.
To reach a complete EF2X allocation we first have as an intermediate goal to reach a (partial) allocation $\X$ satisfying the following properties.

\begin{tcolorbox}[colback=black!5!white,colframe=black!75!black]
\begin{enumerate}
    \item $\X$ is a (partial) EFX orientation.
    \item For any agent $i\neq 0$ and $g \in \U(\X), v_i(X_i) \geq v_i(g)$ and $X_0=\emptyset$.
    \item For any envied agent $i$, $v_i(X_i) \geq v_i\left(\U^*_i(\X)\right)$. 
    \item Every envied agent has at most 1 good.
\end{enumerate}
\end{tcolorbox}

\paragraph{Proof Roadmap.}
We start with an initial allocation (Section~\ref{sec:step1_2X}) satisfying Properties (1)-(2) where each agent other than $0$ receives at most one relevant good: Algorithm~\ref{Algo:Round Robin EF2X} assigns at each agent their most valuable good among the currently unallocated ones. Then, in Algorithm~\ref{Algo:Algo2EF2X} (Section~\ref{sec:step2_2X})
envied agents may switch their allocated good to a set of unallocated relevant goods that they prefer more; those sets are chosen not to include goods relevant to $0$, so EFX does not break for $0$. If such a switch takes place, a further reallocation is needed in order to preserve Property (2).  We note that Algorithms \ref{Algo:Round Robin EF2X} and \ref{Algo:Algo2EF2X}  do not allocate any goods to agent 0, so agent $0$ remains non-envied. After the termination of these algorithms, Properties (3) and (4) are further satisfied and the envy may only be reduced.
Once we have these 4 properties we construct a complete EF2X allocation (Section~\ref{sec:step3_2X}). 
Algorithm~\ref{Algo:FinalAllocationEF2X} carefully assigns each remaining unallocated good to one of the following agents: a non-envied agent for whom the good is relevant, an envied agent for whom the good is relevant but holds only one other good, thereby ensuring EF2X, or agent $0$.

\smallskip
We next  turn to multi-hypergraph instances of multiplicity 2 and prove the second main result of the section, stated below. 

\begin{restatable*}{thm}{ThmEFtXmulti}
    \label{thm:girth3EF3X}
    Instances on multi-hypergraphs of girth at least 3 and multiplicity at most 2 with additive valuations always admit a polynomial-time constructible EF3X allocation. 
    
\end{restatable*}

To reach an EF3X allocation, we have to ensure that the (partial) allocation $\X$ that is reached satisfies the following properties:

\begin{tcolorbox}[colback=black!5!white,colframe=black!75!black]
\begin{enumerate}
    \item $\X$ is a (partial) EF2X orientation.
    \item For any agent $i \neq 0$, and good $g \in \U(\X): v_i(X_i) \geq v_i(\{g, g_t\})$ and $X_0=\emptyset$.
    \item For any envied agent $i$, $v_i(X_i) \geq v_i\left(\U^*_i(\X)\right)$. 
    \item For every good $g$, $g\in \U(\X)$ iff $g_t \in \U(\X)$. 
    \item Every agent receives pairs of twin goods; every envied agent receives a single pair of twin goods.
\end{enumerate}
\end{tcolorbox}

\paragraph{Proof Roadmap.}
The core of this result is similar to above, with the main difference to be that we offer the whole pair of twin goods. The only case that we may break the pair is in the final allocation for goods relevant to agent $0$, which is crucial to achieve EF3X; otherwise, we could guarantee an EF4X allocation as a corollary of Theorem~\ref{thm:girth3} (see Corollary~\ref{cor:ofThm1}).

We start with an initial allocation which satisfies Properties (1)-(2) (Section \ref{Section:step1EF3X}): Each agent is given an arbitrary priority, and  based on that priority we offer each of them to get their most valued unallocated pair of twin goods, except agent 0. 
Afterwards, in Algorithm \ref{Algo:Algo2EF3X} (section \ref{Section:step2EF3X}) for each envied agent $i$ we offer its relevant goods that are irrelevant to agent $0$.  
This is needed to satisfy Property (3). Again, in the case such a change occurs,
a further reallocation is needed in
order to preserve Property (2).
In all these procedures, Properties (4)-(5) are also satisfied.
 Finally, Algorithm \ref{Algo:FinalAllocationEF3X} allocates any remaining goods. Regarding the unallocated pairs of twin goods that are irrelevant to agent $0$ are given to some (arbitrarily chosen) non-envied incident agent, if there exists any, otherwise, they are given to agent $0$. Properties (2)-(3) guarantee the EF2X condition. Regarding the goods relevant to agent $0$ they are oriented appropriately, so that every envied agent receives at most 3 goods, and no envied agent envies agent $0$ (this is done by Property (3) and additivity); non-envied agents do not envy agent $0$ by Property (2).

\subsection{EF2X allocation for hypergraphs with girth at least 3} 
\label{sec:EF2X-simple}
We go on to present the details of the steps that lead to the proof of Theorem \ref{thm:girth3}, restated below.

\ThmEFtXsimple

We prove this theorem in the following three steps.

\subsubsection{Step 1: An EFX orientation such that each envied agent has at most one good}
\label{sec:step1_2X}

In this step every agent other than $0$ gets its most valuable unallocated good.

\begin{algorithm}[h]
\caption{Single Round Robin}
\label{Algo:Round Robin EF2X}
\raggedright\textbf{Input:} A hypergraph $\Hg = (V,E)$ of girth at least 3. \\ 
\textbf{Output:} An allocation $\X$ satisfying properties (1), (2) and (4).
\begin{algorithmic}[1]
\FOR{$i = n-1$ down to $0$ }
    \STATE $X_i\gets \emptyset$
\ENDFOR
\FOR{$i=n-1$ down to 1}
    \STATE $X_{i} \gets  \arg\max_{g \in \U_i(\X)} v_{i}(g)$
\ENDFOR

\end{algorithmic}
\end{algorithm}

\begin{lemma}
    Algorithm~\ref{Algo:Round Robin EF2X} outputs an allocation satisfying Properties (1), (2), (4).
\end{lemma}
\begin{proof}
Properties (1) and (4) are trivially satisfied since in $\X$ every agent receives a single good. Property (2) is satisfied, since when Algorithm~\ref{Algo:Round Robin EF2X} updates the bundle of an agent $i\neq 0$, $i$ is allocated its most valuable unallocated good, and thus it prefers $X_i$ to any unallocated good. Also, agent $0$ is not allocated any good, i.e., $X_0=\emptyset$.
\end{proof}

\subsubsection{Step 2: Reduce Envy.}
\label{sec:step2_2X}

Here,  we offer envied agents all their adjacent, irrelevant to $0$, unallocated goods in place of their current bundle. If an agent
prefers it, it is allocated
that bundle instead. Then, repeatedly, any
unallocated single good is offered to the agents until no agent prefers
any unallocated good.

\begin{algorithm}
\caption{Reducing Envy Algorithm}
\label{Algo:Algo2EF2X}
\raggedright\textbf{Input:} $\X$ satisfying Properties (1), (2) and (4). \\
\textbf{Output:} An allocation $\X$ satisfying Properties (1)-(4).
\begin{algorithmic}[1]
\WHILE{ there exists an envied agent $i$ s.t. $v_i(\U^*_i(\X)) > v_i(X_i)$}
\STATE $X_i \gets \U^*_i(\X)$

    \WHILE{there exists $g \in \U(\X)$ s.t. for some agent $j \neq 0:v_j(g) > v_j(X_j)$}
        \STATE $X_j \gets \arg\max_{g \in \U_j(\X)} v_j(g)$

    \ENDWHILE

\ENDWHILE
\end{algorithmic}
\end{algorithm}

A first thing to note is that whenever some $X_i$ changes  due to the allocating steps of lines 2 or 4, the value that agent $i$ gets strictly increases. Since the bundles (and thus the values) that an agent may get are finite, this directly implies that Algorithm~\ref{Algo:Algo2EF2X} terminates.

\begin{lemma}
    Algorithm~\ref{Algo:Algo2EF2X} outputs an allocation satisfying Properties (1)-(4).

\end{lemma}
    
\begin{proof}
   Regarding Properties (2) and (3), the condition of the outer while-loop is the negation of Property (3) and the condition of the inner while-loop is the negation of Property (2). Thus, Property (2) and Property (3)  hold for the allocation returned by Algorithm \ref{Algo:Algo2EF2X}.
    
    We next show that Property (1), that initially holds, will not break.
    All the updates at lines 2 and 4 orient goods towards their endpoints and thus $\X$ will remain an orientation. Moreover, the allocation remains EFX because envied agents are allocated only one good, which further means that Property (4) is preserved. To see this consider any agent $i$ assigned $\U^*_i(\X)$ at line 2 (which is the only way that an agent is assigned more than one good). Since the allocation is an orientation, $i$ receives at most one good relevant to any other agent (Observation~\ref{obs:at-most-1-relevant-common-edge}), by definition of $\U^*_i(\X)$, all the goods allocated to $i$ are irrelevant to $0$. Suppose that good $g \in \U^*_i(\X)$ is relevant to some other agent $j \neq 0$. Since Property (2) is satisfied before the update of $i$'s bundle (due to the inner while-loop and to the fact that Property (2) was initially satisfied), $j$ does not prefer $g$ to its current bundle at the time of the update and afterwards (since the agents' values may only increase), hence $j$ does not envy $i$. 
\end{proof}

\subsubsection{Step 3: Final allocation}
\label{sec:step3_2X}

The final allocation considers 2 cases for the unallocated goods. First, consider any unallocated good $g$ irrelevant to agent $0$. If $g$ contains any {\em non-envied} agent, $g$ is allocated arbitrarily to any of them, otherwise $g$ is allocated to $0$. Second, consider any unallocated good $g$ relevant to agent $0$, then if $g$ contains at least one {\em envied} agent, $g$ is allocated arbitrarily to any of them, otherwise it is allocated to agent $0$.

\begin{algorithm}[H]
\caption{Final allocation}
\label{Algo:FinalAllocationEF2X}
\raggedright\textbf{Input:} $\X$ satisfying Properties (1)-(4).\\
\textbf{Output:} An EF2X allocation. 
\begin{algorithmic}[1]

\WHILE{ $\exists$ an unallocated good $g$}
    \IF{$g$ is irrelevant to $0$, and $\exists$ a non-envied agent $i \neq 0$ that $g$ is relevant to $i$}
        \STATE $X_i \gets X_i \cup \{g\} $
    \ELSIF{$g$ is relevant to $0$, and $\exists$ an envied agent $i \neq 0$ that $g$ is relevant to $i$}
        \STATE $X_i \gets X_i \cup \{g\} $
    \ELSE
        \STATE $X_0 \gets X_0 \cup \{g\} $
    \ENDIF
        
\ENDWHILE
\end{algorithmic}
\end{algorithm}

\begin{observation}\label{obs:envied-have-only-two}
    In the allocation produced by Algorithm~\ref{Algo:FinalAllocationEF2X}, each envied agent is allocated at most two goods.
\end{observation}
\begin{proof}
    Any envied agent $i$ may only be allocated extra goods in line 5, and those are relevant to agent $0$. By Observation \ref{obs:at-most-1-relevant-common-edge}, only one good can be relevant to both $i$ and $0$, hence $i$ is allocated at most one extra good. 
\end{proof}

\begin{lemma}
    Algorithm~\ref{Algo:FinalAllocationEF2X} outputs an EF2X allocation.
\end{lemma}
\begin{proof}

    We will show that Algorithm~\ref{Algo:FinalAllocationEF2X} does not create any further envy, i.e., no non-envied agent becomes envied. Then, by Observation \ref{obs:envied-have-only-two} the lemma follows.

    We only need to consider lines 3 and 7, and show that the updated bundle is not envied. Consider any non-envied agent $i \neq 0$ that its bundle is updated at line 3. By Observation \ref{obs:at-most-1-relevant-common-edge}, any other agent $j$ for whom $g$ is relevant, has value $0$ for the bundle that $i$ had before, and due to Property (2), $j$ will not envy $i$ after the update. 
    
    Regarding agent $0$, whose bundle may be updated at line 7, there are two cases to consider: $g$ is irrelevant to agent $0$, and only envied agents may find it relevant or $g$ is relevant to agent $0$ and only a non-envied agent may find it relevant. In the first case, note that agent $0$ does not receive any good relevant to both agent $0$ and some envied agent $i$, due to the if-statement at line 4. Therefore, due to Property (3) on the allocation $\X$ that Algorithm~\ref{Algo:FinalAllocationEF2X} receives as input, no envy is created towards agent $0$, even if agent $0$ receives the whole $\U^*_i(\X)$. In the second case, due to Property (2), allocating $g$ to agent $0$ does not cause any envy from agents that find $g$ relevant. 
\end{proof}

\subsubsection{Complexity of constructing the allocation}
\begin{lemma}
    An EF2X allocation for hypergraphs with girth at least 3 can be constructed in polynomial time.
\end{lemma}
\begin{proof}
    We will argue that Algorithms \ref{Algo:Round Robin EF2X}, \ref{Algo:Algo2EF2X}, \ref{Algo:FinalAllocationEF2X} have polynomial time complexity on the number of agents $n$.

    Algorithm \ref{Algo:Round Robin EF2X} gives each agent its most valued unallocated good. Due to Observation \ref{obs:at-most-1-relevant-common-edge}, each agent has at most $n$ relevant goods. Therefore, finding the good with maximum value for each agent requires to check at most $n$ goods, hence Algorithm \ref{Algo:Round Robin EF2X} has time complexity $O(n^2)$.

    Algorithm \ref{Algo:Algo2EF2X} does not create any further envy, and at each loop of the outer while-loop one envied agent becomes non-envied. 
    Hence, each agent can update its bundle at line 2, at most once. Regarding the inner while-loop, each agent may update its bundle by receiving a single good at most $n-1$ times (due to Observation \ref{obs:at-most-1-relevant-common-edge}). In total, each agent may update its bundle at most $n$ times, so the time complexity is $O(n^2)$.

    Concerning Algorithm \ref{Algo:FinalAllocationEF2X}, the outer while-loop can be executed at most $O(n^2)$ times (which is the maximum number of edges/goods in the graph), while the procedure of checking if a good is relevant to an envied agent requires $O(n)$ time complexity, in total the time complexity is $O(n^3)$.

    Overall, an EF2X allocation can be constructed in  polynomial time complexity and specifically in time $O(n^3)$.
\end{proof}

We remark that for multi-hypergraphs of girth at least 3 and multiplicity $k\geq 1$, if we handle all goods of the same set of agents together as a single good and construct an EF2X allocation as above, then each envied agent would hold at most $2k$ goods, resulting in an EF2kX allocation. We present this as a corollary of Theorem \ref{thm:girth3}.

\begin{corollary}
   \label{cor:ofThm1} 
   Instances on multi-hypergraphs of girth at least 3 and multiplicity at most $k$ with general monotone valuations always admit a polynomial-time constructible EF2kX allocation.
\end{corollary}

In the following section we show that specifically for $k=2$ we can get a guarantee of EF3X (in polynomial time) instead of EF4X. 

\subsection{EF3X allocations for multi-hypergraphs with girth at least 3 and multiplicity at most 2}\label{sec:ef2x_multi}
\label{sec:EF2X-multi}

We proceed to a more generalized result in cases of multi-hypergraphs, with multiplicity at most 2. In this section, we present the details of the steps that lead to the proof of Theorem \ref{thm:girth3EF3X}, restated below.
 
\ThmEFtXmulti

\subsubsection{Step 1: An EF2X orientation such that each envied agent has at most two goods}\label{Section:step1EF3X}

We generalize the approach from the previous section, we offer each agent an unallocated good and its twin good together; we always preserve that twin goods are either both allocated or both unallocated (Property (4)).

\begin{algorithm}[H]
\caption{\FixPropMultTwo}
\label{Algo:FixProp123}
\raggedright\textbf{Input:} An allocation $\X$ which satisfies Properties (1), (4), (5).\\
\textbf{Output:} An allocation $\X$ satisfying Properties (1), (2), (4) and (5).
\begin{algorithmic}[1]
\WHILE{there exist $g,g_t \in \U(\X)$ s.t. for an agent $i \neq 0:v_i(\{g, g_t\}) > v_i(X_i)$} 

    \STATE $X_i \gets \arg\max_{g, g_t \in \U(\X)} v_i(\{g,g_t\})$ 

\ENDWHILE

\end{algorithmic}
\end{algorithm}

\begin{lemma} \label{lem:FixProp123}
    Algorithm \ref{Algo:FixProp123} terminates and outputs an allocation which satisfies Properties (1), (2), (4) and (5).
\end{lemma}
\begin{proof}
    The algorithm terminates, since every change in the allocation increases the Social Welfare, and since this is upper bounded by a finite value the algorithm must terminate.

    Property (1) is preserved because at each loop of the while-loop, some agent is oriented 2 goods, so even if it is envied, EF2X is preserved, and that agent only increases its value, so it remains EF2X satisfied. Property (4) is satisfied after each loop of the while-loop since Property (5) was satisfied before: Property (5) guarantees that the set of goods released are only pairs of twin goods, and also the updated set includes a pair of twin goods. This also proves the preservation of Property (5). Finally, the condition of the while-loop is the negation of Property (2) and hence the algorithm returns an allocation satisfying Property (2).  
\end{proof}

\subsubsection{Step 2: Reducing Envy}\label{Section:step2EF3X}

Due to Property (2), any unallocated goods $g, g_t$ irrelevant to agent $0$  can be allocated to any incident non-envied agent, if there exists any, without causing further envy, so without violating EF2X. At this step (in Algorithm \ref{Algo:Algo2EF3X}), we perform reallocations so that we succeed essential conditions that allow us to also allocate all unallocated goods that are relevant only to envied agents by preserving EF2X. Algorithm \ref{Algo:Algo2EF3X} follows the basic idea of Algorithm \ref{Algo:Algo2EF2X}: we offer $\U^*_i(\X)$ to any envied agent $i$, and if it accepts, we perform the more special updates of \FixPropMultTwo($\X$) (Algorithm~\ref{Algo:FixProp123}) that considers allocating twin goods together to preserve Property (2) as well.

\begin{algorithm}[H]
\caption{Orienting hyperedges}
\label{Algo:Algo2EF3X}
\raggedright\textbf{Input:} A multi-hypergraph $\Hg = (V,E)$ of girth at least 3 and multiplicity at most 2.  \\
\textbf{Output:} An allocation $\X$ satisfying Properties (1)-(5).
\begin{algorithmic}[1]

\STATE Let $\X$ be the all empty allocation.
\STATE $\X \gets$ \FixPropMultTwo($\X$)

\WHILE{ there exists an envied agent $i$ s.t. $v_i(\U^*_i(\X)) > v_i(X_i)$}
    \STATE $X_i \gets \U^*_i(\X)$
    \STATE $\X \gets$ \FixPropMultTwo($\X$)
\ENDWHILE

\end{algorithmic}
\end{algorithm}

\begin{lemma}
    Algorithm \ref{Algo:Algo2EF3X} terminates, and outputs an allocation which satisfies Properties (1)-(5).
\end{lemma}

\begin{proof}
    The algorithm terminates, because whenever an agent changes its bundle, the Social Welfare increases, and the Social Welfare is upper bounded by a finite amount.
    The empty allocation trivially satisfies Properties (1), (4) and (5), since there are no envied agents. After the execution of \FixPropMultTwo\ in line 2, due to Lemma \ref{lem:FixProp123}, the resulting allocation satisfies all properties but possibly Property (3). Those properties are satisfied at the beginning of each loop of the while-loop, as long as Properties (1), (4), (5) are preserved after the update of line 4, which is what we show next, since then   \FixPropMultTwo\ in line 5 would guarantee again all properties but possibly Property (3).

    First, we show that no agent envies an agent $i$ that changed its bundle in line 4.
    This is true for all agents but $0$ due to Property (2), which is satisfied at the beginning of each loop of while-loop (due to the execution of \FixPropMultTwo\ in lines 2 and 5). This is also trivially true for agent $0$, since the updated set allocated to $i$ is irrelevant to agent $0$.
    Moreover, after line 4 agent $i$ will be EF2X-satisfied towards any other agent as Property (1) was satisfied before, and $i$ only increases its value with this update.  

    Further note that Properties (4) and (5) are preserved after the update in line 4: Property (4) was satisfied before the update, meaning that $\U^*_i(\X)$ in line 4 includes only pairs of twin goods, preserving that way Property (5). On the other hand, Property (5) was  satisfied before the update, meaning that the goods released are pairs of twin goods, preserving that way Property (4). 
    
    Overall, 
    the execution of \FixPropMultTwo\ in line 5, guarantees that all properties but possibly (3) are satisfied at the beginning of each loop of the while-loop, and also after the termination of Algorithm \ref{Algo:Algo2EF3X}. 
    Finally, Property (3) is also satisfied after the termination of Algorithm \ref{Algo:Algo2EF3X} since it is the negation of the while-loop condition.

\end{proof}

\subsubsection{Step 3: Final Allocation}\label{Section:step3EF3X}

In this final step, we allocate any unallocated goods such that no further envy is created. Where to allocate a good is decided by two main factors: if it is relevant to agent 0 and how many envied agents find relevant that good.

In high level, we orient goods irrelevant to $0$ to non-envied agents, and then allocate $\U^*_i(\X)$ to agent $0$ for any envied agent $i$. Due to Properties (2)-(3) these allocations would not cause any further envy. For any unallocated good $g$ relevant to agent $0$, if there is a single incident envied agent  $i$, we offer $i$ its most valuable good between $g$ and its twin good $g_t$. Property (3) combined with the additivity of the goods guarantee that $i$ will not envy agent $0$, and Property (2) guarantees that nobody else envies agent $0$. For any other unallocated good $g$ relevant to $0$, if there are incident envied agents, they receive a single good between $g$ and $g_t$, otherwise, if there is no incident envied agents, agent $0$ receives both twin goods $g$ and $g_t$. Due to Property (5), we guarantee that any envied agent receives at most 3 goods. This is sufficient to show that the allocation is EF3X.

\begin{algorithm}
\caption{Final allocation}
\label{Algo:FinalAllocationEF3X}
\raggedright\textbf{Input:} $\X$ satisfying Properties (1)-(5).\\
\textbf{Output:} An EF3X allocation.
\begin{algorithmic}[1]

\FOR{every good $g \in \U^*(\X)$} 
\IF{ $\exists$ non-envied agent $i$ incident to $g$}
    \STATE $X_i \gets X_i \cup \{g, g_t\}$
\ELSE
    \STATE $X_0 \gets X_0 \cup \{g, g_t\}$
\ENDIF
\ENDFOR 

\FOR{every good $g \in \U_0(\X)$}
    \IF{ There are no envied agent $i$ incident to $g$}
        \STATE $X_0 \gets X_0 \cup \{g,g_t\}$
    \ELSIF{$\exists$ a single envied agent $i$ incident to $g$}
        \STATE w.l.o.g., let $v_i(g) \geq v_i(g_t)$ 
        \STATE $X_i \gets X_i \cup \{g\}$
        \STATE $X_0 \gets X_0 \cup \{g_t\}$
    \ELSE 
        \STATE Allocate each of $\{g,g_t\}$ to a different envied agent. \COMMENT{In this case, there are at least 2 envied agents}
    \ENDIF
\ENDFOR 
                
\end{algorithmic}
\end{algorithm}

\begin{lemma}
    Algorithm \ref{Algo:FinalAllocationEF3X} terminates and outputs an EF3X allocation.
\end{lemma}
\begin{proof}
    The algorithm terminates as it allocates at each iteration at least one good. We examine each possible augmentation done in different lines of the algorithm to show that the allocation returned by Algorithm \ref{Algo:FinalAllocationEF3X} satisfies the Lemma's conditions.
    We will either show that the augmentation creates no envy towards non-envied agents, which is sufficient for the desired allocation due to Property (1) not breaking towards those agents, or that for any agent $i$ that is envied is allocated at most one extra good, meaning that at the end it receives at most 3 goods (due to Property (5)) and EF3X is satisfied.
    
    We first argue about lines 1-7, where only irrelevant goods to agent 0 are allocated. Each allocation in lines 3 and 5 does not create any envy towards the non-envied agents which are given a pair of twin goods goods. To see this, in line 3 no envy is created towards $i$ due to Property (2). In line 5 no envy is created towards agent $0$ due to Property (3) as agent $0$ receives a pair of twin goods, $\{g,  g_t\}$, only when all agents that find them relevant are envied.

    Lines 8-18 deal with the remaining goods which all are relevant to agent 0. Line 10 does not create any envy as Property (2) holds. In line 13 agent $i$ is allocated an additional good, since Property (5) holds for agent $i$ before allocating $g$, then by allocating $g$, any agents that envied $i$ are going to be EF3X-satisfied; note that agent $i$ shares at most one pair of twin goods with agent $0$ (Observation \ref{obs:repetition-common-edge}), so it is considered at most once in the second for-loop. In line 14 all incident non-envied agents, or in other words, all incident agents  apart from $i$, do not envy agent 0 due to Property (2). We also argue that agent $i$ does not envy agent $0$. Assume $\X$ is the input allocation to Algorithm~\ref{Algo:FinalAllocationEF3X}. Since Property (3) holds in $\X$, $v_i(X_i) \geq v_i(\U^*_i(\X))$. After the update in $i$'s bundle in line 13, $i$'s bundle becomes $X_i\cup\{g\}$. By additivity, and the assumption w.l.o.g. that $v_i(g)\geq v_i(g_t)$ (line 12), it holds that $v_i(X_i\cup\{g\}) \geq v_i(\U^*_i(\X)\cup\{g\}) \geq v_i(\U^*_i(\X)\cup\{g_t\})$. Note, further, that at the end of Algorithm~\ref{Algo:FinalAllocationEF3X}, the only goods allocated to agent $0$ that are relevant to agent $i$ form a subset of $\U^*_i(\X)\cup\{g_t\}$, and as a result, agent $i$ does not envy agent $0$.

    In line 16 goods are allocated to envied agents, and we will show that each of them may receive at most one extra good. Note that any envied agent $i$ may receive extra goods only in lines 13 and 16, and each good $g$ that $i$ receives should be relevant to both agents $i$ and $0$ and furthermore its twin good $g_t$ is given to another agent (agent $0$ in line 14 or some other envied agent in line 16). However, by Observation \ref{obs:repetition-common-edge} agents $i$ and $0$ have only $g$ and $g_t$ as common relevant goods, which guarantees that only one such good may be given to $i$. Since Property (5) holds for agent $i$ before allocating $g$, then by allocating $g$, any agent that envies $i$ is going to be EF3X-satisfied.
\end{proof}

\subsubsection{Complexity of constructing the allocation}
\begin{lemma}
    An EF3X allocation for hypergraphs with girth at least 3 and multiplicity at most 2 can be constructed in polynomial time.
\end{lemma}
\begin{proof}
    We will argue that Algorithms \ref{Algo:FixProp123}, \ref{Algo:Algo2EF3X}, \ref{Algo:FinalAllocationEF3X} have polynomial time complexity on the number of agents, let it be $n$.
    First, observe that for each agent there are at most $n$ possible alternatives to receive a bundle of a relevant good and its twin good (Observation \ref{obs:at-most-1-relevant-common-edge}).

    Algorithm \ref{Algo:FixProp123}: Whenever an agent changes its bundle, it increases its value, so each agent can change its bundle at most $n$ times, giving a total number of $O(n^2)$ updates. Additionally, the check of the while-loop condition needs $O(n^2)$ time complexity ($n$ agents and $O(n)$ possible sets of goods). So the total complexity of Algorithm \ref{Algo:FixProp123} considering all possible calls by Algorithm \ref{Algo:Algo2EF3X} is $O(n^4)$.
    
    Algorithm \ref{Algo:Algo2EF3X}: Algorithm \ref{Algo:Algo2EF3X} needs $O(n^4)$ time for the total calls of Algorithm \ref{Algo:FixProp123}, as we argued above. We next argue that there will be at most $n^2$ loops of the while-loop. 
    The reason is that after executing a loop of the while-loop, which means that an agent $i$ violates the condition of the  while-loop, that agent becomes non-envied due to Property (2). That agent may become envied again during the execution of Algorithm \ref{Algo:FixProp123} in line 5, where it receives a pair of twin goods. Since there are at most $n$ such pairs that are relevant to $i$ and its value only increases, agent $i$ may be considered in the while-loop at most $n$ times. So, in total, there may be at most $n^2$ loops of the while-loop. Checking the condition of the while-loop can be executed by checking for each agent if it is envied that needs $O(n)$ time, and compare its bundle with $\U^*_i(\X)$ that needs $O(1)$ time. So, checking the while-loop condition needs $O(n^2)$ time. The condition of the if-statement can be checked in $O(1)$ time by checking if the goods relevant to both agents $i$ and $0$ were unallocated prior to the update of $i$'s bundle. Overall, the time complexity of Algorithm \ref{Algo:Algo2EF3X} after excluding the complexity of Algorithm \ref{Algo:FixProp123} is $O(n^4)$, and all the calls of Algorithm \ref{Algo:FixProp123} need $O(n^4)$ time in total. Therefore, the overall time complexity of Algorithm \ref{Algo:Algo2EF3X} is $O(n^4)$.  

    Algorithm \ref{Algo:FinalAllocationEF3X}: The two for-loops combined are going to be executed $O(n^2)$ times (which is the maximum number of goods), the procedure of checking if a good has at least two, exactly one, or zero incident envied agents requires $O(n^2)$ time complexity (for each agent we check if any other agent envies it), and the procedure of checking if a good is allocated can be done in $O(1)$ time complexity. In total the time complexity is $O(n^4)$.
    
     Since the sum of the time complexity of Algorithms \ref{Algo:FixProp123}, \ref{Algo:Algo2EF3X}, \ref{Algo:FinalAllocationEF3X} is polynomial, we can construct an EF3X allocation in polynomial time.
\end{proof}

\section{Pursuing Approximate EFX Allocations}\label{sec:apxEFX}
In this section, we study two types of instances. We begin with hypergraph instances whose underlying hypergraph has girth at least 3 and the valuation functions are subadditive. We then extend our analysis to a more general setting involving multi-hypergraph instances, where the underlying hypergraph has girth at least 3 and multiplicity 2, while the valuation functions are assumed to be additive.

\subsection{Further preliminaries on valuations and goods}
Before we state the main results of that section, we give two extra definitions on the types of goods that are crucial throughout the constructions and the algorithms that will follow. 

\noindent{\bf Critical Goods}.
If $X_i\subseteq M$ is the set of goods allocated to some agent $i \in N$, then a good $g \notin X_i$ is {\em critical} for $i$, if  $g$ has more than half the value of $X_i$, in $i$'s perspective, i.e., $v_{i}(g) > \frac{1}{2}v_{i}(X_{i})$.

\noindent{\bf Dummy Goods}. A good is called \textit{dummy} if it is irrelevant to all agents. In our algorithms we introduce a set $H=\{h,h_1,\ldots,h_n\}$ of $n+1$ dummy goods, and by this we mean that the new set of goods is $M\cup H$, $h$ is irrelevant to all agents and for every $i,j\in N$: $v_j(h_i)=0$   and $h_i$ is artificially considered  relevant only for agent $i$. 

\begin{remark}
Dummy goods are there to help keeping track of the number of goods the agents get during the executions of the algorithms. It will be easy to verify that a good $h$ is {\em never} allocated  to some agent
and  that if some  agent $i\in N$ receives her dummy good $h_i$, then she also receives exactly one more good, valuable for her. 
\end{remark}

In the algorithms presented, for most of the updates of the allocations we will use the virtual valuations $\hvv=(\hv_1,\hv_2,\ldots,\hv_n)$ defined below. We note that these are similar to the factors used in the potential function of \citet[Definition 13]{hv2025almost}, and the construction of the envy graph in \citet[Definition 2.3]{filosratsikas2026approximateenvyfreeallocationsk}.

\noindent {\bf Virtual Valuation Functions}. For some $0<\alpha<1$, for every agent $i$ with valuation function $v_i$ we define the \textit{virtual valuation function}  $\hv_i:2^{M}\rightarrow \mathbb{R}_{\geq 0}$ so that 
\begin{equation*}
    \hv_i(S) = 
     \begin{cases}
      v_i(S), & \text{if }|S|\leq 1\\
       \frac{1}{\alpha}v_i(S), & \text{if }|S|\geq 2\
     \end{cases}
\end{equation*}
In the algorithms and in the procedures, when we consider virtual valuations, we define, in a similar way as before, the virtual envy and the virtual envy graph as follows. \\
\noindent{\bf Virtual Envy}. For some $0<\alpha<1$ we say that an agent $i$ {\em virtually} envies another agent $j$ (or bundle $X_j$), if $\hv_{i}(X_{i}) <\hv_{i}(X_{j})$. \\ 
\noindent{\bf Virtual Envy Graph}. Given a partial allocation $\X$ and the virtual valuation functions of the agents we define the corresponding virtual envy graph  $\G(\X) = (\N,\tilde{E}(\X))$ to be a directed graph where its vertex set is the set of agents and its edge set contains a directed edge $(i,j)$ if and only if  agent $i$  {\em virtually} envies agent $j$, i.e., $\tilde{E}(\X) = \{(i,j):\hv_{i}(X_{i}) <\hv_{i}(X_{j})\}$.

\noindent{\bf Potential Function}. In order to prove that the algorithms we present terminate in pseudo-polynomial time we make use of the following potential function $\Phi$ that is based on virtual valuations.
$\Phi$ maps allocations to the non-negative  real numbers so that for a (partial) allocation $\X$: 
$$\Phi(\X)=\sum_{i=1}^n\hv_i(\X_i)$$

\noindent To handle some extreme cases we need to build on the above potential function and employ a lexicographic potential, making use of the lexicographic order of vectors of $\mathbb{R}^2$: ($a,b)<(c,d)$ if and only if  $a<b$ or [$a=b$ and $c<d$]. We will use it on vectors of the form $\Big(\Phi(\X),\sum_{i=1}^n|X_i|\Big)$, where $\X$ is an allocation. 

Note that, since the goods and all the possible allocations are finite, $\Phi(\X)$ and $\Big(\Phi(\X),\sum_{i=1}^n|X_i|\Big)$ are upper bounded.

\subsection{Results and roadmap}

For hypergraph instances, we reprove in a simpler way a theorem originally established by \citet{kaviani2024envyfreeallocationindivisiblegoods} as stated below, with extra justification on the  running time complexity. 

\begin{restatable*}{thm}{ThmHyperOne}
\label{thmhyper1}
Instances on hypergraphs of girth at least 3 
under subadditive valuations admit a $\frac{\sqrt{2}}{2}$-EFX allocation that can be computed in pseudo-polynomial time.
\end{restatable*}

For its proof we rely on Algorithm \ref{alg:g3m1}, which outputs such an allocation. It begins by finding a partial allocation  satisfying the following property, for any $0<\alpha<1$.
\begin{tcolorbox}[width=6in, title= %
Property $\prpt$ of a partial allocation $\X$]
$\X$ is an $\alpha-$EFX allocation, where, additionally, every agent $i \in N$ with $|X_{i}| \leq 1$ is EFX towards any other agent.
\end{tcolorbox}
\noindent In the proofs, whenever we want to prove that Property $\prpt$  holds for the allocation at hand, we may split the proof in proving separately the following two properties:
\begin{tcolorbox}[width=6in, title= Properties $\prpt$(a) and $\prpt$(b) of a partial allocation $\X$]
  \begin{itemize}
     \item[\textbf{$\prpt$(a):}] Every agent $i \in N$ with $|X_{i}| \leq 1$ is EFX towards any other agent.
     \item[\textbf{$\prpt$(b):}] Every agent $i \in N$ with $|X_{i}| \geq 2$ is $\alpha-$EFX towards any other agent.
       \end{itemize}
\end{tcolorbox}
Note that Property $\prpt$ holds iff Properties $\prpt$(a) and $\prpt$(b) hold. \\ 

\paragraph{Proof Roadmap.}
Our approach is an extension of Algorithm 1 of \cite{amanatidis2024pushingfrontierapproximateefx}. To reach a partial allocation that satisfies Property $\prpt$, the algorithm initially allocates bundles of goods  to agents (Algorithm \ref{algo:prePro}) so that  no subset of these bundles is envied by any agent and also the pool is not envied by any agent, where the envy is based on the virtual valuations. Then, taking advantage of Observation \ref{obs:at-most-1-relevant-common-edge} and that the pool is not envied, it allocates to every agent with more than one goods, her relevant goods from the pool.

If at this point there exists a source with two goods in the envy graph (based on the virtual valuations), it could assign the remaining goods of the pool to that source and get a complete allocation. To make sure this is the case, it gives to each source with one good, its dummy good and then  resolves all the cycles of the envy graph (Algorithm \ref{algo:postPro}), improving the agents' virtual values. Yet, when the cycles are resolved the pool changes and envy towards the pool may be created. To control the envy towards the pool it re-initiates the initial allocating procedure, until eventually (since the agents' virtual values increase) it finds such a source and the pool is not envied by any agent.  Based on Property $\prpt$ and that no agent envies the pool, by setting $\alpha=\frac{\sqrt{2}}{2}$,  the complete allocation is shown to be $\frac{\sqrt{2}}{2}-$EFX

\smallskip
For multi-hypergraph instances we get the following theorem.

\begin{restatable*}{thm}{ThmHyperTwo}\label{thm:hyper2}
Instances on multi-hypergraphs of girth at least 3 and multiplicity 2 under additive valuations, admit a  $\frac{2}{3}$-EFX allocation that can be computed in pseudo-polynomial time.
\end{restatable*}

\paragraph{Proof Roadmap.} For its proof we rely on Algorithm \ref{algo:hyper2}, which outputs such an allocation. The approach is similar to before using again Algorithms \ref{algo:prePro}   and \ref{algo:postPro} as subroutines to reach a partial allocation where  the envy towards the pool is controlled and Property $\prpt$ is satisfied.  Here, for the envy towards the pool, the algorithm makes sure that no agent with at least two goods has a critical good in the pool. 

Once such an allocation is found the algorithm would wish to assign all unallocated goods to the sources of the virtual valuations' envy graph so that no subset of the sources' bundles  is envied by any agent. So, in a ``testing step'',  it starts allocating unallocated goods to sources until all of them become envied. If the pool becomes empty by this procedure then by the properties of the partial allocation, by setting $\alpha =\frac{2}{3}$, the resulting complete allocation is shown to be $\frac{2}{3}$-EFX. If the pool does not become empty and there are no sources to allocate goods from the pool, then the algorithm allocates goods only to the sources belonging to an arbitrarily chosen cycle in the envy graph (created during  the testing step) and ignores all the allocations of goods to other sources of this step.  It then resolves this cycle, improving agents' valuations, restores  Property $\prpt$ (that may have been broken) by removing irrelevant goods from bundles, and goes on from the start until eventually (since the agents' virtual values increase) in some testing step the pool gets empty. 

\smallskip

We close the section by proving a lemma that will be used in the forthcoming proofs. It shows that if for some allocation $\X$,  Property $\prpt$ holds, then changing the bundle of an agent to a better one based on the virtual valuations,  Property $\prpt$ does not break for that agent. 
\begin{lemma}
\label{lem:Prop_Preserve}
     Let $\X$ be an allocation for which Property $\prpt$ holds  
     and let agent $i$ and set of goods $S$ be such that $\hv_i(S)\geq\hv_i(X_i)$. If  $i$ gets assigned bundle $S$, for any agent $j$:
     \begin{itemize}
         \item[(a)] if $|S|\leq 1$, $i$ will be EFX towards $X_j$, and
         \item[(b)] if $|S|\geq 2$, $i$ will be  $\alpha$-EFX towards $X_j$.
     \end{itemize}
\end{lemma}

\begin{proof}
    If $|X_j|=1$ then trivially $i$ is EFX towards $j$. It remains to prove the lemma for when $|X_j|\geq 2$.

    Regarding (a), suppose $|S|\leq1$. If $|X_i|\leq 1$ then by the definition of the $\hv_i$'s: $\hv_i(S)>\hv_i(X_i)\Rightarrow v_i(S)>v_i(X_i)$.   Moreover,  by  Property $\prpt$(a) of $\X$, it is $v_i(X_i)\geq v_i(X_j\setminus\{g\})$ for any $g\in X_j$, and thus $v_i(S)\geq v_i(X_j\setminus\{g\})$ for any $g\in X_j$, as needed. 
   If $|X_i|\geq2$ then by the definition of the $\hv_i$'s: $\hv_i(S)>\hv_i(X_i)\Rightarrow v_i(S)>\frac{1}{\alpha}v_i(X_i)$. Moreover,  by  Property $\prpt$(b) of $\X$, it is $v_i(X_i)\geq \alpha v_i(X_j\setminus\{g\})$ for any $g\in X_j$, and thus $v_i(S)\geq v_i(X_j\setminus\{g\})$ for any $g\in X_j$, as needed. 

    Regarding (b), suppose $|S|\geq2$. If $|X_i|\leq 1$ then by the definition of the $\hv_i$'s: $\hv_i(S)>\hv_i(X_i)\Rightarrow \frac{1}{\alpha}v_i(S)>v_i(X_i)$.   Moreover,  by  Property $\prpt$(a) of $\X$, it is $v_i(X_i)\geq v_i(X_j\setminus\{g\})$ for any $g\in X_j$, and thus $v_i(S)\geq \alpha v_i(X_j\setminus\{g\})$ for any $g\in X_j$, as needed. 
    If $|X_i|\geq 2$ then by the definition of the $\hv_i$'s: $\hv_i(S)>\hv_i(X_i)\Rightarrow \frac{1}{\alpha}v_{i}(S) > \frac{1}{\alpha}v_{i}(X_{i}) \Rightarrow  v_i(S)>v_i(X_i)$.  Moreover,  by  Property $\prpt$(b) of $\X$, it is $v_i(X_i)\geq \alpha v_i(X_j\setminus\{g\})$ for any $g\in X_j$, and thus $v_i(S)\geq \alpha v_i(X_j\setminus\{g\})$ for any $g\in X_j$, as needed. 
\end{proof}

\subsection{Basic subroutines}
We will provide subroutines for  reducing envy towards the pool, for eliminating cycles in  envy graphs and for  reducing envy between agents.

\subsubsection{A Pre-Processing  Procedure}
Here, we introduce a pre-processing procedure (Algorithm \ref{algo:prePro}) that is used in our constructions. Intuitively,  Algorithm \ref{algo:prePro} iteratively swaps a subset of the pool with an agent's allocated bundle, as long as it prefers that subset according to the virtual valuations and $\prpt$ is preserved. To do this, we employ an algorithm of Chaudhury et al. \cite[Algorithm 2.3]{CKMS21}, given that at least one agent virtually envies the pool, the algorithm finds a subset $T$ of the pool that is virtually envied by some agent $j$ and no other agent virtually envies any proper subset of $T$. Algorithm \ref{algo:prePro} allocates $T$ to agent $j$ and $j$'s  old bundle is released. This is repeated until no agent, based on the virtual valuations, prefers any subset  of  the pool. 

We present for completeness the algorithm of Chaudhury et al. \cite[Algorithm 2.3]{CKMS21} in the Appendix. We call this algorithm \MinEnviedSet\ that receives as input an allocation $\X$ and an agent $i$ that virtually envies $\U(\X)$, and returns the desirable set $T\subseteq \U(\X)$ and an agent $j$, such that $j$ virtually envies $T$, and no agent virtually envies any proper subset of $T$.

\begin{algorithm}
\caption{Procedure1($\X$)}
\label{algo:prePro}
\raggedright{\bf{Input}:}An orientation $\X$ satisfying Property $\prpt$.\\

{\bf{Output:}} An orientation $\X$ satisfying Property $\prpt$.\\ 

\begin{algorithmic}[1]
\WHILE{$\exists$ $i\in N$ such that $\hv_i(\U(\X))>\hv_i(X_i)$ }
    \STATE $(T,j) \gets$ \MinEnviedSet $(\X,i)$ \COMMENT{$j$ virtually envies $T$, and no agent virtually envies any $S\subset T$}
    \IF {$|T|=2$ and $T$ contains an irrelevant good $r$ for $j$} 
    \STATE $T\gets (T\setminus\{r\})\cup \{h_j\}$
    \ENDIF
    \STATE  $X_{j} \leftarrow T$
\ENDWHILE
\RETURN $\X$
\end{algorithmic}
\end{algorithm}

First thing to note is that the potential function strictly increases whenever the allocation changes inside the while-loop and, eventually,  Algorithm \ref{algo:prePro} terminates. This is because at every  step that the allocation changes (line 6), 
the virtual value of the bundle that agent $j$ gets, strictly increases in her perspective, and there can be finitely many possible such  values for every agent. Thus, at some point there should not exist any agent that gets more value from any subset of the pool implying that the while-loop of the algorithm will break. The above yields the following observation. 

\begin{observation}
\label{obs:pre}
    Algorithm \ref{algo:prePro} terminates and returns an allocation where none of the virtual values of the agents has decreased and no agent virtually  envies the entire (or any subset) of the  pool, i.e., for any $i\in N$, $\hv_i(X_i)\geq \hv_i(\U(\X))$. Additionally, whenever the allocation inside the while-loop of Algorithm \ref{algo:prePro} changes from say an $\X$ to some $\X'$ it will be $\Phi(\X)<\Phi(\X')$.   
\end{observation}

\begin{lemma}
\label{lem:propPreserve_pre}
If an orientation $\X$ satisfies  Property $\prpt$ then the output of Procedure1($\X$) is an orientation that satisfies Property $\prpt$.
\end{lemma}

\begin{proof}
We first prove that the allocation remains an orientation. Subset $T$ of line 2 has to contain at least one relevant good for agent  $j$ in order to have $\hv_j(T)> \hv_{j}(X_{j})$ in line 2. If it contains more than one relevant goods, then, due to its minimality, it cannot contain any irrelevant  goods. If it contains exactly one relevant good, say $g$, then, due to its minimality, it may contain  at most one irrelevant good, say $r$. This $r$ would be there only to adjust the cardinality of $T$ to $|T|=2$, since it could be $\hv_\hi(g)=v_\hi(g)\leq \hv_\hi(X_\hi)<\frac{1}{\alpha}v_\hi(g)=\frac{1}{\alpha}v_\hi(\{g, r\})=\hv_\hi(\{g, r\})$.  In words, it may be that, based on the virtual values, $g$ alone is not preferred by agent $\hi$ but $g$ together with any other (irrelevant) good  is preferred. In place of $r$ we use the (artificially) relevant good $h_\hi$ and thus, in this case as well, agent $\hi$ will get only relevant goods in the assignment of line 6. 

We now turn to inductively prove that the allocating step of line 6 will not break Property $\prpt$. To do so, we will first examine the way that agent $\hi$ values the bundles of the rest of the agents, and then  examine the way that other agents value the new bundle assigned to $\hi$. For the rest of the agents' pairs, any agent other than $\hi$ will value the bundle of any agent other  than $\hi$ in the same way as before and thus Property  $\prpt$ cannot break because of them.

We start by examining the way that $\hi$, based on the original valuations, values the bundles of the rest of the agents after the allocating  step of line 6.
Let $i$ be some other agent. 
Assuming that $\prpt$ is satisfied at the beginning of each loop of the while-loop, the requirements of Lemma \ref{lem:Prop_Preserve} are met after the allocating step of line 6 and thus the envy of agent $\hi$ towards other agents' bundles  cannot be a reason for Property $\prpt$ to break.

Let us now examine the way that other agents value the new bundle assigned to $\hi$, i.e., $T$. Consider any other agent $i$.
If $|T|=1$, then $i$ will trivially be EFX (and $\alpha$-EFX) towards $\hi$. If $|T|\geq 2$, for any strict subset $S\subset T$ it is $\hv_i(S)\leq \hv_i(X_{i})$ (by the conditions hold for the output of \MinEnviedSet). If $|S|\geq 2$ or  $|X_i|\leq 1$, by the definition of the $\hv_j$,  this gives $v_i(S)\leq  v_i(X_{i})$  and thus $i$ will be EFX (and $\alpha$-EFX) towards $\hi$. If $|S|\leq 1$ and $|X_i|\geq 2$ then $\hv_i(S)\leq \hv_i(X_{i})$ translates to $v_i(S)\leq \frac{1}{\alpha}v_i(X_{i})\Leftrightarrow \alpha v_i(S)\leq v_i(X_{i})$
and so $i$ will be $\alpha$-EFX towards $\hi$. Thus, no agent's envy towards $\hi$ can be a reason for Properties $\prpt$(a) and $\prpt$(b), and therefore $\prpt$, to break.
\end{proof}

\subsubsection{Eliminating Cycles}

The subroutine CycleResolution eliminates any given cycle in the envy graph, by giving to  every agent in the cycle the bundle that the next agent in the cycle has. AllCyclesResolution repeatedly calls CycleResolution until all cycles in the corresponding envy graph are eliminated. Both subroutines terminate in polynomial time.

\begin{algorithm}
\caption{CycleResolution$(\X,\vv,G(\X),C)$}
\label{algo:cycle_res}
\raggedright{\bf{Input}:} An allocation $\X$, its envy graph $G(\X)$ according to $\vv$, and a cycle $C$ in $G(\X)$.

\raggedright{\bf{Output}:} An updated allocation $\X$ such that its envy graph $G(\X)$ according to $\vv$ does not contain $C$.

\begin{algorithmic}[1]
     \STATE $\X' \leftarrow \X$ 
    \FOR{every edge $(i,j) \in C$}
        \STATE $X_i \leftarrow X'_j$
    \ENDFOR
    \STATE \textbf{return} $\X$
\end{algorithmic}
\end{algorithm}

\begin{algorithm}
\caption{AllCyclesResolution$(\X,\vv)$}
\label{algo:all_cycle_res}
\raggedright{\bf{Input}:} An allocation $\X$ and the valuation profile of the agents $\vv$.

\raggedright{\bf{Output}:} An updated  allocation $\X$ such that the envy graph $G(\X)$ according to $\vv$ is acyclic. 

\begin{algorithmic}[1]
        \WHILE{there exists a cycle $C$ in the envy graph $G(\X)$ according to $\vv$} 
        \STATE $\X$  \ = \  CycleResolution$(\X,\vv, G(\X),C)$
    \ENDWHILE
    \RETURN $\X$
\end{algorithmic}
\end{algorithm}

AllCyclesResolution (and CycleResolution) will be called to resolve cycles in $\G(\X)$, i.e., the envy graph based on the virtual valuations. The following lemma is important for our proofs.

\begin{lemma}
\label{lem:Prop_Preserve_cycle}
    If an allocation $\X$ satisfies Property $\prpt$, then the output of AllCyclesResolution$(\X,\hvv)$ satisfies Property $\prpt$.
\end{lemma}  

\begin{proof}
    First note that in any execution of CycleResolution, some bundles go from one agent to another but their contents do not change. Thus,  agents that do not participate in  the resolved cycle keep their bundles and continue envying the same bundles the way they did before. Thus Property $\prpt$ cannot break because of the way agents outside the cycle see other agents' bundles. It remains to show that this is true for agents in the resolved cycle.

    Consider an agent $i$ in the cycle that will be resolved
    and the bundle $X_j$ of some agent $j$ different from $i$. When the cycle is resolved we know that agent $i$ will get the set that the next agent in the cycle has, say  $S$, for which $\hv_i(S)>\hv_i(X_i)$. But then, for $\X$, $i$, $S$ and $j$ the requirements of Lemma \ref{lem:Prop_Preserve} are met and thus the envy of $i$ towards other agents' bundles cannot be a reason for Properties $\prpt$(a) and $\prpt$(b) to break, irrespective of where these bundles will be allocated after the cycle resolution.
\end{proof}

\subsubsection{A Post-Processing Procedure}

The following post-processing procedure (Algorithm \ref{algo:postPro}) is used to resolve cycles with the ultimate goal of finding a source in the envy graph having either zero or at least two goods in her bundle. For that, we make use of the dummy goods  $\{h,h_1,\ldots,h_n\}$ , that will help in controlling  the cardinalities of the bundles. Recall $\G$ denotes the envy graph  that arises when one considers the virtual valuations.

\begin{algorithm}
\caption{Procedure2($\X$)}
\label{algo:postPro}
\raggedright{\bf{Input}:} An orientation $\X$ satisfying Property $\prpt$.\\

{\bf{Output:}} An orientation $\X$ satisfying Property $\prpt$, with $\G(\X)$ being acyclic. 

\begin{algorithmic}[1]
\WHILE {$\G(\X)$ has cycles or some source of $\G(\X)$ has only one good} 
    \FORALL {sources of $\G(\X)$}
    \STATE If for a source $s$ it is $|X_s|=1$ then $X_s\leftarrow X_s\cup \{h_s\}$
    \ENDFOR
    \STATE   $\X \leftarrow$ AllCyclesResolution$(\X,\hvv) $ \COMMENT{Update $\X$ by eliminating all cycles}
    \STATE Let $S = \{i \in N| \ |X_{i}| \geq 2\}$
    \FORALL {$i\in N$}
    \STATE Remove from $X_i$ all the goods that are irrelevant to $i$ 
    \ENDFOR 
    \FORALL {$i\in S$ so that $|X_i|=1$}
    \STATE $X_i\leftarrow X_i\cup\{h_i\}$ 
    \ENDFOR     
\ENDWHILE
\RETURN $\X$
\end{algorithmic}
\end{algorithm}

First thing to note is that  at every execution of the while-loop at least one of lines 3 and 5 will change the allocation so that at least some of the virtual values of the agents' bundles 
get strictly increased, while the rest remain unchanged: line 3 increases the virtual value of every source $s$ because of the cardinality of $X_s$ switching from 1 to 2,\footnote{This is true only if $s$ had positive value for its assigned good. If this is not the case, still the termination is guaranteed, since there are at most $n$ sources.} line 5  increases the virtual values of the bundles of the agents in the cycle (if there is a cycle in $\G(\X)$), and lines 8 and 11 keep the same values for all agents and for those with more than two goods keep the cardinality above one, thus keeping the virtual values the same. As a result, since there can be finitely many possible such  values for every agent, Algorithm \ref{algo:postPro} terminates. In other words, the condition of the outer while-loop gets satisfied and we have the following observation.

\begin{observation}
\label{obs:post}
    Algorithm \ref{algo:postPro} terminates returning an allocation where none of the virtual values of the agents has decreased and for which $\G(\X)$ is acyclic and for any source $s$ of $\G(\X)$ it is $|X_s|=0$ or $|X_s|\geq 2$. Additionally, whenever the allocation changes after the execution of a loop of the   while-loop of Algorithm \ref{algo:postPro} from say an $\X$ to some $\X'$  the lexicographic potential increases, i.e., it will be either $\Phi(\X)<\Phi(\X')$, or it is $\Phi(\X)=\Phi(\X')$ and $\sum_{i=1}^n|X_i|<\sum_{i=1}^n|X'_i|$.  
\end{observation}

We also have the following crucial lemma.

\begin{lemma}
\label{lem:propPreserve_post}
If $\X$ is an orientation that satisfies Property $\prpt$, then the output of Procedure2($\X$)  is an orientation that satisfies Property $\prpt$.
\end{lemma}

\begin{proof}
    Allocation $\X$ is initially an orientation and  will remain an orientation no matter how many times the while-loop will be executed, since in the last run of the while-loop, line 8 will discard from every $X_i$ all irrelevant goods for $i$, and line 11 will  add only $h_i$ to  $X_i$, which is (artificially) relevant  for $i$.

    The allocation changes in lines 3, 5, 8 and 11. We will show that in none of them  Property $\prpt$ can break. In line 3, for any such agent $s$ and any agent $i$, it is $v_i(X_s)=v_i(X_s\cup\{h_s\})$ and thus the way agents envy each other does not change. Similarly, with the changes in line 11 the way agents envy each other does not change. In line 5, AllCyclesResolution is called but, because of Lemma \ref{lem:Prop_Preserve_cycle}, Property $\prpt$ cannot break there.  In line 8, the (original) value of the bundle of any agent $i$ does not change in $i$'s perspective since only irrelevant goods for $i$ are removed, whereas, on the other hand  the rest of the bundles may loose value in $i$'s perspective. Thus, Property $\prpt$ cannot break due to line 8 as well.
\end{proof}

\subsection{Solving hypergraph instances of girth at least 3}{\label{sec:hyper_g_3}}

We next present an algorithm for computing a $\frac{\sqrt{2}}{2}$-EFX allocation for hypergraph instances of girth at least three. Here we will (eventually) use $\alpha=\frac{\sqrt{2}}{2}$. Starting with  all bundles empty, the algorithm repeatedly and iteratively takes improving steps with regard to the envy towards the pool and the envy among agents, until it reaches an allocation where the envy towards the pool and among the agents simultaneously satisfies  some suitable properties. In the final step it assigns all the unallocated goods to some (arbitrarily chosen) source of $\G(\X)$ and reaches a $\frac{\sqrt{2}}{2}$-EFX complete allocation.

In more details, after initializing all agents' bundles to empty and introducing a set of dummy goods (lines 1-4), in the outer while-loop (lines 5-13) the algorithm repeatedly eliminates the envy  towards the pool (via calling Algorithm \ref{algo:prePro} in line 7), allocates to agents with more than one goods their relevant goods from the pool (lines 8-10) and resolves the possibly created cycles on the virtual envy graph $\G(\X)$ (via calling Algorithm \ref{algo:postPro} in line 11), and keeps on doing so as long as one of these steps change the allocation (line 12). Once this while-loop stops, an arbitrary source of $\G(\X)$ is chosen and  gets allocated all the unallocated goods (lines 14-15). The resulting allocations is returned (line 17) after removing all allocated dummy goods from it (line 16).

\begin{algorithm}
\caption{Hyper1 Algorithm}
\label{alg:g3m1}
\raggedright\textbf{Input:} A hypergraph instance $\Hg = (V,E)$ of girth at least 3.
    
    \raggedright\textbf{Output:} A $\frac{\sqrt{2}}{2}$-EFX complete allocation $\X$ satisfying Property $\prpt$.
\begin{algorithmic}[1]
    \FOR{every agent $i \in N$}
    \STATE $X_{i} \leftarrow \emptyset$
    \ENDFOR
    \STATE Introduce set $H=\{h,h_1,\ldots,h_n\}$ of $n+1$ dummy goods. 
    \WHILE{True}
        \STATE  $\X'\leftarrow \X$
        \STATE $\X \leftarrow $ Procedure1($\X$). 
        \WHILE{there is an agent $i \in N$ with $|X_{i}| \geq  2$ and a good $g \in \U(\X)$ such that $v_i(g)>0$}
        \STATE  $X_{i} \leftarrow X_i\cup \{g\}$ \COMMENT{Agent i receives all her relevant valuable goods from the pool}
        \ENDWHILE
         \STATE  $\X \leftarrow $ Procedure2($\X$).
         \STATE  If $\X = \X'$, break;
    \ENDWHILE \\ 
\STATE Let $s \in N$ be a source of $\G(\X)$
\STATE  $X_{s} \leftarrow X_{s} \cup \U(\X)$ \COMMENT{Assigning the pool to a non-envied source $s$}
\STATE Remove all dummy goods from the agents'  bundles
\RETURN $\X$
\end{algorithmic}
\end{algorithm}

Again, a first thing to note is that Algorithm \ref{alg:g3m1} terminates. To prove this it suffices to show that the while-loop of lines 5-13 terminates, or, in other words, that the allocation at some run of the while-loop will remain the same (line 12). This is true because whenever the allocation changes inside the while-loop, it does so either because of lines 7 or 11, where the  lexicographic potential increases (Observations~\ref{obs:pre} and \ref{obs:post}), or because of line 9 which can be executed a finite number of times and only increases the virtual values that agents perceive for their bundles while increasing some bundles' sizes, thus increasing the lexicographic potential. Since there can be finitely many possible values for the lexicographic potential, Algorithm \ref{alg:g3m1} terminates. 

Moreover, once the while-loop of lines 5-13 terminates, we know that in the last run of the  while-loop the allocation  did not change, which implies that the calls of Procedure1 at line 7 and Procedure2 at line 11 did not change the allocation. In turn this implies that  Observations \ref{obs:pre} and \ref{obs:post} will simultaneously hold for the allocation $\X$ at hand. Last thing to note is that with the allocating steps of the while-loop of lines 8-10 any $i$ with $|X_i|\geq 2$ has no relevant good in the pool. All the above yield the following.

\begin{observation}
    \label{obs:hyper1}
    Once the while-loop of lines 5-13 of Algorithm \ref{alg:g3m1} terminates, for  allocation $\X$   Observations \ref{obs:pre} and \ref{obs:post} hold, and for any $i$ with $|X_i|\geq 2$ it holds that  $v_i(\U(\X))=0$.
    Additionally, whenever the allocation changes after the execution of a loop of the   while-loop of lines 5-13 of Algorithm \ref{alg:g3m1} from say an $\X$ to some $\X'$ it will be $\Big(\Phi(\X),\sum_{i=1}^n|X_i|\Big)<\Big(\Phi(\X'),\sum_{i=1}^n|X'_i|\Big)$.  
\end{observation}

Below is the main result of the section.

\ThmHyperOne

\begin{proof}
We will show that Algorithm \ref{alg:g3m1} returns such an allocation.  We defer the discussion on the running time of the algorithm to Section \ref{sec:RunTimes} (Lemma \ref{lem:runTime1}). We start by showing that when the while-loop of lines 5-13 terminates, allocation $X$ is an orientation that satisfies Property $\prpt$. When all bundles are empty  Property $\prpt$ trivially holds and the allocation is an orientation. Also, by Lemmas \ref{lem:propPreserve_pre} and \ref{lem:propPreserve_post} when calling  Procedure1 and Procedure2 in lines 7 and 11, respectively,  Property $\prpt$ cannot break and the allocation remains an orientation. Last, in the while-loop of lines 8-10 agents are given only goods relevant to them. Thus, the allocation is at all times an orientation and it suffices to show that in the while-loop of lines 8-10 Properties  $\prpt$(a) and   $\prpt$(b) do not break. 

Consider agent $i$, with $|X_i|\geq 2$, that will get an extra good because of line 9, and any other agent $j$. It is $v_i(X_i)\leq v_i(X_i\cup\{g\})$, for the $g$ allocated in line 9, and thus $i$ will not stop being $\alpha$-EFX towards $j$, implying that Properties $\prpt$(a) and $\prpt$(b) cannot break because of the way $i$ values other agents' bundles. It remains to examine how other agents value $X_i\cup\{g\}$. 

 Consider any agent $j$ and note that $X_i$ contains only goods relevant to $i$ (recall, $\X$ is at all times an orientation). If $X_i$ contains the only common  good  of $i$ and $j$ (Observation \ref{obs:at-most-1-relevant-common-edge} for girth$\geq$3 instances), then for $j$ it is $v_j(X_i)= v_j(X_i\cup\{g\})$, implying that $j$ will keep envying $i$ in the same way as before (EFX or $\alpha$-EFX towards $i$).
 If $X_i$ does not contain the only common  good of $i$ and $j$, then it might still be in the pool. 
 Yet, by every termination of Procedure1 at line 7, due to Observation \ref{obs:pre}, for any agent $j$: $\hv_j(X_j)\geq \hv_j(\U(\X))$, which, by the definition of the $\hv_j$'s, implies $\hv_j(X_j)\geq v_j(g)$. Moreover, this inequality holds throughout the while-loop of lines 8-10, since the agents' values only increase and goods are only removed from the pool. On the other hand, since $X_i$ contains only goods relevant to $i$  and does not contain the common good of $i$ and $j$, it is $v_j(g)=v_j(X_i\cup\{g\})$.  If $|X_j|=1$, then $\hv_j(X_j)\geq v_j(g) \Rightarrow v_j(X_j)\geq v_j(g)$, yielding $v_j(X_j)\geq v_j(X_i\cup\{g\})$, implying that $j$ will be  EFX towards $i$. If $|X_j|\geq 2$, then $\hv_j(X_j)\geq v_j(g) \Rightarrow \frac{1}{\alpha}v_j(X_j)\geq v_j(g)$, yielding $v_j(X_j)\geq \alpha v_j(X_i\cup\{g\})$, implying that $j$ will be  $\alpha$-EFX towards $i$.

Since by the termination of the while-loop of lines 5-13 the allocation at hand satisfies Properties $\prpt$(a) and $\prpt$(b),  every agent is $\alpha$-EFX towards any other agent. In order to complete the proof it suffices to show  that with the allocating step of line 15 all agents are $\alpha$-EFX towards the chosen source $s$. For any agent $i$ with $|X_{i}| \geq 2$, from Observation \ref{obs:hyper1},  it holds that $v_{i}(\U(\X)) = 0$. On the other hand since $s$ from line 14 is a source of $\G(\X)$  we have that $\hv_{i}(X_{i}) \geq \hv_{i}(X_{s})$ and by Observation \ref{obs:post}, $|X_s|=0$ or $|X_s|\geq 2$. If $|X_s|\geq 2$ then $\hv_{i}(X_{i}) \geq \hv_{i}(X_{s})\Rightarrow v_{i}(X_{i}) \geq v_{i}(X_{s})$. The last inequality trivially holds if $|X_s|=0$. Putting it all together, by the subadditivity of the valuation functions, we get $$v_{i}(X_{s} \cup \U(\X)) \leq v_{i}(X_{s}) + v_{i}(\U(\X)) = v_{i}(X_{s}) \leq v_{i}(X_{i}) $$

\noindent and so $i$ will not envy $s$ after the allocating step of line 15, trivially  being also $\alpha$-EFX towards $s$.

For any agent $i \in N $ with $|X_{i}| \leq 1 $, by Observation \ref{obs:pre} it holds that $\hv_i(X_i)\geq \hv_i(\U(\X))$. If $i$ is allocated some good then it should be of positive value which implies that $h,h_1\in \U(\X)$. Thus, $|\U(\X)|\geq 2$ 
and we get  $\hv_i(X_i)\geq \hv_i(\U(\X))\Rightarrow v_i(X_i)\geq \frac{1}{\alpha}v_i(\U(\X))$. 
On the other hand, since $s$  of line 14 is a source  of $\G(\X)$ it is $\hv_i(X_i)\geq \hv_i(X_s)$, which if $|X_s|\geq 2$  implies $v_i(X_i)\geq \frac{1}{\alpha}v_i(X_s)$, with the last inequality being trivially true if $|X_s|=0$ (by Observation~\ref{obs:post} there is no other case for $|X_s|$). By the subadditivity of the valuation functions we have that 
$$v_{i}(X_{s} \cup \U(X)) \leq v_{i}(X_{s}) + v_{i}(\U(X)) \leq a\cdot v_{i}(X_{i}) + a\cdot v_{i}(X_{i}) = 2a\cdot v_{i}(X_{i})$$
and thus agent $i$ will be $\frac{1}{2\alpha}-$EFX towards $s$ after the allocating step of line 15. Setting $\alpha =\frac{\sqrt{2}}{2}$ we get $\frac{1}{2\alpha}=\alpha$ and  thus agent $i$ will be $\alpha-$EFX towards $s$ after the allocating step of line 15. 

Putting it all together, for $\alpha =\frac{\sqrt{2}}{2}$, $\X$ is a complete allocation where  every agent is $\alpha$-EFX towards any other agent.
\end{proof}

\subsection{Solving multi-hypergraph Instances of girth at least 3 and multiplicity 2}
\label{sec:multi-hyper} 
In this section we present an algorithm for computing a $\frac{2}{3}-$EFX allocation for multi-hypergraph instances of girth at least 3 and multiplicity 2.  Here we will (eventually) use $\alpha=\frac{2}{3}$. The algorithm starts with all bundles empty and  takes improving steps aiming  to reduce the envy towards the pool and among the agents but also to release all  critical goods from the pool. The ultimate goal is to reach a final updated suitable allocation where certain properties are satisfied which will be shown to be a $\frac{2}{3}-$EFX complete allocation.

In more detail, after initializing all agents' bundles to empty and introducing a set of dummy goods (lines 3-6), the algorithm  enters the while-loop of lines 7-40. There, among other things,  a testing step is run that determines when this while-loop breaks (lines 24-27, and we will prove it will eventually break). In this while-loop, until it breaks, the algorithm repeatedly does the following:

In the while-loop (lines 7-21) it repeatedly eliminates the envy  towards the pool (via calling Algorithm \ref{algo:prePro} in  line 10), carefully allocates all critical goods from the pool to agents depending  on the agents' bundles' sizes and where the twin goods of the critical goods (if existent) are allocated (lines 11-18) and resolves the possibly created cycles on the virtual envy graph $\G(\X)$ (via calling Algorithm \ref{algo:postPro} in line 19), and keeps on doing so as long as one of these steps change the allocation (line 20).  Once this while-loop stops (we will show it has to),  it initiates a testing step that from $\X$ it creates an allocation $\X'$ by allocating one by one unallocated goods to sources of the virtual envy graph  until either the pool gets empty or, because of this testing step,  there are no more sources in the virtual envy graph (while-loop of lines 24-26). 

If the procedure stopped with an empty pool, then the outer while-loop of lines 7-40 breaks (line 27) and the resulting allocation is returned  after removing all dummy goods from it (lines 41-42).  If the procedure stopped with a non-empty pool but with no more sources in the virtual envy graph, then an arbitrary cycle of the virtual envy graph is chosen (line 28) and the sources of this cycle  are allocated the bundles they got in the testing step, i.e., the bundles they get in $\X'$ (lines 29-31), whereas all other agents still get the same bundles as right before the initialization of the testing step, i.e., the bundles they had in $\X$. Under this new allocation the arbitrary chosen cycle exists in the virtual envy graph and the algorithm resolves it (line 32). This run of the while-loop of lines 7-40 ends by removing all irrelevant goods from the agents' bundles (line 34-37) making sure, by using the dummy goods, that agents with more than one goods will still have more than one goods (lines 33 and 37-39). Then, the next iteration of the while-loop starts and this goes on until eventually (as we will prove) one iteration ends with the corresponding testing step leaving the pool empty.

\begin{algorithm}[t]
\caption{Hyper2 Algorithm}
\label{algo:hyper2}
\begin{algorithmic}[1]
    \STATE \textbf{Input:} A multi-hypergraph instance $\Hg = (V,E)$ of girth at least  3 and multiplicity 2.
    \STATE \textbf{Output:} A $\frac{2}{3}$-EFX complete allocation $\X$ satisfying Property $\prpt$.\\ 
     \FOR{every agent $i \in N$}
    \STATE $X_{i} \leftarrow \emptyset$
    \ENDFOR 
    \STATE Introduce set $H=\{h,h_1,\ldots,h_n\}$ of $n+1$ dummy goods.
    \WHILE{True}
        \WHILE{True}
        \STATE  $\X'\leftarrow \X$
            \STATE $\X \leftarrow $ Procedure1($\X$). \\
    \WHILE{there is an  agent $i \in N $ and a good $g \in \U(\X)$ such that $v_{i}(g) > \frac{1}{2}v_{i}(X_{i})$ }
    \STATE {\bf(i)} if $|X_{i}| =  2$ and there is a  $g'\in X_{i}$  such that $v_{i}(g) > v_{i}(g')$ then $X_{i} \leftarrow (X_{i} \cup  \{g\}) \setminus \{g'\}$ \\ 
    \STATE {\bf(ii)} if $|X_{i}| >2$ and $g_t$ is the twin  good of $g$, if it exists,  then\\
\STATE \qquad  If  $g_{t} \notin X_{i}$  then  $X_{i} \leftarrow X_{i} \cup \{g\}$  \\ 
\STATE \qquad  If  $g_{t} \in X_{i}$ then \\ 
\STATE \qquad \qquad if  $v_{i}(g_{t}) \leq \frac{1}{2}v_{i}(X_{i})$ then  $X_{i} \leftarrow (X_{i} \cup \{g\}) \setminus \{g_{t}\}$ \\ 
\STATE \qquad \qquad else if $v_{i}(g_{t}) >\frac{1}{2}v_{i}(X_{i})$ then $X_{i} \leftarrow \{g,g_{t}\}$. \\ 
\ENDWHILE
\STATE $\X \leftarrow $ Procedure2($\X$).\\
\STATE  If $\X = \X'$, break;
\ENDWHILE \\ 
    \STATE Let $\So(\X)$ be the set of sources in $\G(\X)$
    \STATE $\X' \leftarrow \X$ \COMMENT{$\X'$ is a potential update for agents $s \in \So(\X)$, which will either break the outer while-loop (line 27) or lead to the update of line 30.} 
    \WHILE{$\U(\X')\setminus H \neq \emptyset$ and there exists a source  $s\in \So(\X)$ that is not virtually envied under $\X'$} 
    \STATE  $X'_{s} \leftarrow X'_{s} \cup \{g\}$, for some  $g\in \U(\X')\setminus H$, by giving priority to relevant goods for $s$.
    \ENDWHILE
    \STATE If  $\U(\X') \setminus H= \emptyset$ then  $\X \leftarrow \X'$, break; 
    \STATE Else there exists a cycle $C \in \G(\X')$, and let $S_{C} = C\cap \So(\X) \neq \emptyset$, be the sources of $\G(\X)$ that participate in $C$. 
    \FORALL{sources $s\in S_C$}
    \STATE $X_s\leftarrow X'_s$. \COMMENT{We keep the update only for agents $s\in S_{C}$}
    \ENDFOR
    \STATE  $\X \leftarrow $ CycleResolution$(\X,\hvv,\G(\X),C)$ 
    \STATE Let $S = \{i \in N| \ |X_{i}| \geq 2\}$
    \FORALL {$i\in N$}
    \STATE Remove from $X_i$ all the goods irrelevant to $i$ 
    \ENDFOR 
    \FORALL {$i\in S$ so that $|X_i|=1$}
    \STATE $X_i\leftarrow X_i\cup\{h_i\}$
    \ENDFOR 
   
    \ENDWHILE \\
    \STATE Remove all dummy goods from the agents'  bundles
\RETURN $\X$ 
\end{algorithmic}
\end{algorithm}

First we note that Algorithm \ref{algo:hyper2} terminates, or, in other words, that the outer while-loop of lines 7-40 terminates. For that, similar to before, we will rely on the fact that at every run of the outer while-loop, for any $0<\alpha<1$,   the virtual value of at least one agent strictly increases without any decrease on other agents' virtual values, which, since the possible virtual values are finite, implies that the algorithm will terminate (more accurately, the lexicographic potential strictly increases). We will show that this is true separately for lines 8 to 21 and lines 22 to 39.

\begin{observation}
\label{obs:hyper2_upperHalf}
    The while-loop of lines 8-21 terminates with an allocation where none of the virtual values of the agents has decreased and there are no critical goods in $\U(\X)$ for agents with at least two goods in their bundles.
\end{observation}

\begin{proof}

At any repetition of the while-loop of lines 8-21, by Observations \ref{obs:pre} and \ref{obs:post}, the procedures of lines 10 and 19 terminate without any agent getting a bundle with less virtual value, and in the while-loop of lines 11-18, the value of every agent $i$ that gets its bundle changed strictly increases: in all the allocating steps of lines 12-17 a critical good $g$ with $v_i(g)>\frac{1}{2}v_i(X_i)$ is either added to $X_i$  as an extra good  or as a replacement of a good with value $\leq\frac{1}{2}v_i(X_i)$, or is allocated to agent  $i$ together with another good that also has value $\geq \frac{1}{2}v_i(X_i)$ (recall, the valuation functions are additive). Since it remains that $|X_i|\geq 2$ this  implies that also the virtual value of any such agent  $i$ strictly increases. Since the possible virtual values are finite, the while-loop of lines 11-18 will eventually  stop executing whenever called and  eventually, for the same reason, so will the while-loop of lines 8-21.
\end{proof}

\begin{observation}
\label{obs:hyper2_lowerHalf}
    Algorithm \ref{algo:hyper2} terminates with a complete allocation.
    Additionally, whenever the allocation changes after the execution of a loop of the   while-loop of lines 7-40 of Algorithm \ref{algo:hyper2} from say an $\X$ to some $\X'$ it will be $\Big(\Phi(\X),\sum_{i=1}^n|X_i|\Big)<\Big(\Phi(\X'),\sum_{i=1}^n|X'_i|\Big)$.
\end{observation}

\begin{proof}
    By Observation~\ref{obs:hyper2_upperHalf}, the while-loop of lines 8-21 terminates with an allocation where none of the virtual values of the agents has decreased. Continuing with lines 22-39, first observe that the while-loop of lines 24-26 terminates, since the pool has finitely many goods, and that  the allocating steps of lines 25 and 30 only increase the value of the source and its bundle's size, thus only increasing the source's virtual value. If the while-loop of lines 24-26 ends with an empty $\U(\X)$ then the outer while-loop will break and the algorithm will terminate (line 27). In different case all the sources of $\G(\X)$ have become envied under $\X'$ and thus in line 32  a cycle will be resolved. There, every agent participating in the cycle that gets resolved will  get a bundle with strictly greater virtual value, in her perspective. Last, in lines 33-39 the agents' virtual values do not decrease since only irrelevant goods are removed from the bundles and all bundles with size $\geq2$ keep this property. Putting it all together, as long as the outer while-loop does not break, the virtual values of the agents do not decrease and at least one agent gets its virtual value increased.  Since the possible virtual values are finite, the algorithm will terminate with $\U(\X)=\emptyset$.
\end{proof}

It remains to show that the allocation returned by Algorithm \ref{algo:hyper2} is a $\frac{2}{3}$-EFX allocation. For that, we first show that as long as the outer while-loop of lines 7-40 does not break,  Property $\prpt$ holds. Once the outer while-loop breaks (line 27), the allocation may not have Property $\prpt$ due to the while-loop of lines 24-26. Yet, up to that point Property $\prpt$ did hold and we will be able to show that at the end, by setting $\alpha=\frac{2}{3}$, the output is a $\frac{2}{3}$-EFX allocation.

We begin by showing, in two steps, that as long as the outer while-loop of lines 7-40 does not break,  Property $\prpt$ holds, and later we will fix $\alpha=\frac{2}{3}$.

\begin{lemma}
\label{lem:propPreserve_hyper2Middle}
Let $\X$ be an orientation that satisfies Property $\prpt$. If $\X$ is the allocation when the while-loop in lines 11-18 starts executing, then upon its termination the resulting allocation will be an orientation that satisfies Property $\prpt$.
\end{lemma}

\begin{proof}
Clearly the allocation remains an orientation since for any agent $i$ considered in line 11, only a critical for $i$, and thus relevant to $i$, good $g$  enters her bundle. To show that Properties $\prpt$(a) and $\prpt$(b) do not break we first show that $i$ keeps being $\alpha$-EFX towards all other agents (observe $|X_i|\geq2$).
Let $i$ with $|X_{i}| \geq 2$ be the agent considered in line 11. In all cases of lines 12, 14, 16 and 17 the critical good $g$ will be allocated to her   either as an extra good (line 14), or as a replacement of a good of smaller value (lines 12 or 16) or as a bundle with another good (i.e., the twin good of g) which also has  value $>\frac{1}{2}v_i(X_i)$ (line 17). Thus, since the valuation functions are additive, in all cases, the value of agent $i$ increases when allocated the new bundle,
implying that Properties $\prpt$(a) and $\prpt$(b) cannot break because of the way $i$ values other agents' bundles. 

It remains  to examine how other agents value the new bundle of $i$.
Consider an agent $j$. When Procedure1 (line 10) terminates, by Observation \ref{obs:pre} it will be $\hv_j(X_j)\geq \hv_j(\U(\X))$. This implies that $v_j(X_j)\geq v_j(\{g,h\})=v_j(g)$. Keep in mind this  last inequality since we use it in all cases below, and also that the valuation functions are additive.

To begin with, let agent $i$ be allocated her new bundle by line 12, where $|X_i|=2$. The good that gets replaced is $ g'$ and let $g''$ be the good that will remain in  $X_i$. In this case the new bundle for $i$ will be $\{g,g''\}$. If $|X_j|\leq 1$, since (inductively) $j$ is EFX towards $i$, it is $v_j(X_j)\geq v_j(X_i\setminus\{g'\})=v_j(g'')=v_j(\{g,g''\}\setminus \{g\})$ which together with $v_j(X_j)\geq v_j(g)=v_j(\{g,g''\}\setminus \{g''\})$ implies that $j$ will be EFX towards $\{g,g''\}$ and thus towards $i$.
If $|X_j|\geq 2$, similar to above, since (inductively) $j$ is $\alpha $-EFX towards $i$, it is $v_j(X_j)\geq \alpha v_j(X_i\setminus\{g'\})=\alpha v_j(g'')$ which together with $v_j(X_j)\geq v_j(g)$ implies that $j$ will remain $\alpha $-EFX towards $i$. 

Next, let  agent $i$ be allocated her new bundle by line 14, where $g$ is added as an extra good to $X_i$. In this case $X_i$ does not contain $g_t$, i.e., the twin of $g$. This, together with  $\X$  (inductivley) being an orientation and that $i$ and $j$ can share at most one good and its twin (Observation \ref{obs:repetition-common-edge}), yields $v_j(X_i)=0$ or $v_j(g)=0$. If $v_j(X_i)=0$ then $v_j(X_i\cup\{g\})=v_j(g)$. Since $v_j(X_j)\geq v_j(g)$, the latter implies that $j$ will not envy $i$ and thus will be EFX (and $\alpha $-EFX) towards $i$. If $v_j(X_i)\neq 0$ then $v_j(g)=0$. Since $|X_i|>2$, $\X$ is an orientation  and $i$ and $j$ share at most one good and its twin, there exist a $g'\in X_i$   irrelevant to $j$ for which $v_j(X_i)=v_j(X_i\setminus\{g'\})$. 
If $|X_j|\leq1$ then $j$ is EFX towards $i$ implying that $v_j(X_j)\geq v_j(X_i\setminus \{g'\})=v_j(X_i)=v_j(X_i\cup\{g\})$,  and thus $j$ will not envy $i$, being EFX towards $i$.
If $|X_j|\geq2$ then $j$ is $\alpha$-EFX towards $i$ implying that $v_j(X_j)\geq \alpha v_j(X_i\setminus \{g'\})=\alpha v_j(X_i)=\alpha v_j(X_i\cup\{g\})$,  and thus $j$ will be $\alpha$-EFX towards $i$.

For the third case let  agent $i$ be allocated her new bundle by line 16, where the twin good $g_t$ of $g$ replaces $g$ in $X_i$.
 If $g$ is relevant to $j$ then also is $g_t$ as twin good of $g$, and, by Observation \ref{obs:repetition-common-edge} and the fact that an orientation is maintained, $g_t$ is the only good in $X_i$ relevant to $j$ and thus, $v_j(X_i\setminus \{g_t\})=0$. On the other hand we have $v_j(X_j)\geq v_j(g)= v_j(X_i\setminus \{g_t\}\cup \{g\})$, implying that $j$ will not envy $i$ and thus will be EFX (and $\alpha $-EFX) towards $i$. 
 If $g$ is irrelevant to $j$ then so is its twin good $g_t$,  and $v_j(X_i\setminus\{g_t\})=v_j(X_i)$. If $|X_i|\leq1$ then (inductively) $j$ is EFX towards $i$ and thus $v_j(X_j)\geq v_j(X_i\setminus\{g_t\})= v_j(X_i)=v_j(X_i\cup\{g\}\setminus\{g_t\})$,  implying that $j$ will remain EFX towards $i$. Similar, if $|X_i|\geq 2$ then (inductively) $j$ is $\alpha $-EFX towards $i$ and thus $v_j(X_j)\geq \alpha v_j(X_i\setminus\{g_t\})= \alpha v_j(X_i)=\alpha v_j(X_i\cup\{g\}\setminus\{g_t\})$,  implying that $j$ will remain $\alpha $-EFX towards $i$.

For the last case, let agent $i$ be allocated her new bundle by line 17. In this case the new bundle for $i$ will be $\{g,g_t\}$. If $g$ and $g_t$ are irrelevant to $j$ then clearly $j$ will not envy $i$ and thus  will be EFX (and $\alpha $-EFX) towards $i$.  If $g$ and $g_t$ are relevant to $j$ then $X_i$ has no other good relevant to $j$ (Observation \ref{obs:repetition-common-edge}), which implies that $v_j(g_t)=v_j(X_i)=v_j(X_i\setminus\{g'\})$, for some (any) $g'\in X_i, g'\neq g_t$ (recall $|X_i|>2$).
If $|X_j|\leq 1$, since (inductively) $j$ is EFX towards $i$, it is $v_j(X_j)\geq v_j(X_i\setminus\{g'\})=v_j(g_t)$ while on the other hand $v_j(X_j)\geq v_j(g)$, and thus, $j$ will be EFX towards $\{g,g_t\}$ and remain EFX towards $i$. Similar, If $|X_j|\geq 2$, since (inductively) $j$ is $\alpha $-EFX towards $i$, it is $v_j(X_j)\geq \alpha v_j(X_i\setminus\{g'\})=\alpha v_j(g_t)$ while on the other hand $v_j(X_j)\geq v_j(g)$, and thus, $j$ will be $\alpha $-EFX towards $\{g,g_t\}$ and remain $\alpha$-EFX towards $i$.

In conclusion, no matter which line will change $i$'s bundle, any agent $j$ with $|X_j|\leq 1$ will remain EFX towards $i$ and any agent $j$ with $|X_j|\geq 2$ wil remain $\alpha$-EFX towards $i$ and thus Properties $\prpt$(a) and $\prpt$(b) will not break because of the way that other agents value the new bundle of $i$. 
 \end{proof}

\begin{lemma}
\label{lem:prop_preserve_fixSources}
Let $\X$ be an orientation that satisfies Property $\prpt$. If $\X$ is the allocation when lines 22-39  starts executing, then if  the outer while-loop does not break in line 27, when line 39 is reached, the resulting allocation will be an orientation that satisfies Property $\prpt$.
\end{lemma}

\begin{proof}

Let us assume that $\X$ is the allocation when lines 22-39 started executing and the outer while-loop did not break in line 27. Then for-loop of lines 34-36 (and that of lines 37-39) makes sure that the allocation when reaching line 39 is an orientation. We go on to prove that Properties $\prpt$(a) and $\prpt$(b) will also hold when reaching line 39.  

Because of Lemma \ref{lem:Prop_Preserve}, we need only to show that these properties do not break because of the envy towards bundles that come from sources.  
To see this implication of Lemma \ref{lem:Prop_Preserve}, let $\X'$ be the allocation when line 39 is reached and for any agent $j$, let $X^*_j$ be the bundle that resulted from the manipulation of $X_j$ up to line 39. For $\X$, any agent $i$, $X_i'$ and any agent $j$ the lemma implies that these properties do not break for the bundles  left unchanged in $\X'$, i.e., $X_j=X^*_j$. Additionally, if $j$ lies in the resolved cycle of line 28 and it is not a source, then its bundle $X_j$, after the cycle's resolution (line 32), will turn to $X_j^*$ by  only getting goods removed  from it (line 35) and possibly one dummy good added to it (line 38). 
Therefore, the properties do not break because of the way $i$ values such $X^*_j$'s as well. The only case left is to examine the envy that $i$ may have towards  $X^*_j$'s that come from sources, i.e., $j$ is a source, and that is what we do next, focusing on sources of $\X$.

If when entering line 22 there was a source $s$ with $|X_{s}| = 0$, then, since no agent virtually envies the pool (Observation \ref{obs:pre}), 
in lines $24-26$ the source could be allocated the entire pool without getting  envied by any agent, and so the outer while-loop would break in line 27. Thus, all sources had at least two goods in $\X$ (Observation \ref{obs:post}).

When the algorithm reaches line 31, only the sources of a single cycle of $\G(\X)$ have their bundle enhanced with extra goods.
Let $s$ be one of these sources. When the cycle gets resolved some agent, say $i$, will get $X_s$, which in lines 33-39 will be manipulated so that it contains only relevant goods for $i$, without losing the property $|X_s|\geq2$. Let this manipulated bundle given to $i$ be $X_i'$. We will show that 
for any $j \neq i$, agent $j$ is EFX towards $X_i'$.

Let $g_\ell$ be the last good added to $X_i'$. By the condition of line 24, agent $j$ may virtually envy $X_i'$ only after the addition of  $g_\ell$  and so $g_\ell$ must be relevant to $j$ for $j$ to envy $X_i'$. On the other hand, since $i$ received $X_s$ in the cycle resolution it should be that she envied $s$, implying, by the same reasoning, that $g_\ell$ is relevant to $i$ as well. Thus for $X_i'$ to be envied by $j$ it should be that $g_\ell$ is one of the at most two goods that $i$ and $j$ may share (Observation \ref{obs:repetition-common-edge}) and so, after removing from $X_s$ the  goods that are irrelevant to $i$, $X_i'$ may contain only two goods relevant to $j$, namely $g_\ell$ and its twin $g_\ell'$ , if existent. We will distinguish between two cases: $g_\ell$ being irrelevant to $s$ and $g_\ell$ being relevant to $s$. 

If $g_\ell$ is irrelevant to $s$ then 
both $g_\ell$ and its twin $g_\ell'$ (if existent) 
belonged to  $\U(\X)$ when line 22 was reached ($\X$ was an orientation and priority is given to relevant goods in line 25) and thus, by Observation \ref{obs:pre},  $\hv_j(X_j)\geq \hv_j(\{g_\ell,g_\ell'\})=\hv_j(X_i')$, implying that $X_i'$ is not envied by $j$. 
If $g_\ell$ is relevant to $s$ then, since $X_s$ contains only goods relevant to $s$ ($\X$ was an orientation and priority is given to relevant goods in line 25), $X_i'$ contains at most two goods (Observation \ref{obs:repetition-common-edge}), namely $g_\ell$ and its twin $g_\ell'$ (if existent). If
both $g_\ell$ and its twin $g_\ell'$ (if existent)  
belonged to  $\U(\X)$, similar to above, when line 22 was reached then $X_i'$ is not envied by $j$. If only  $g_\ell$ belonged to  $\U(\X)$  then it was not envied (by Observation \ref{obs:pre}) but also, if $g_\ell'\in X_s$, since $X_s$ was not envied, $g_\ell'$ is not envied by $j$. Thus, either if $X_i'=\{g_\ell\}$ or  $X_i'=\{g_\ell,g_\ell'\}$, agent $j$ will be EFX towards $X_i'$.
\end{proof}

Below we state the main theorem of this section. 

\ThmHyperTwo

\begin{proof}
We show that Algorithm \ref{algo:hyper2} returns such an allocation. We defer the discussion on the running time of the algorithm to Section \ref{sec:RunTimes} (Lemma \ref{lem:runTime2}). The algorithm terminates soon  after the while-loop of lines 7-40 terminates in line 27, at which point the pool contains only dummy goods. 
Using already proved lemmas we will first argue that whenever the algorithm reaches line 22 the allocation will be an orientation that satisfies Property $\prpt$. Note that Property $\prpt$ implies that the allocation up to that point will be  $\alpha$-EFX, for any $0<\alpha<1$. Then, using this property for $\alpha=\frac{2}{3}$ and that no agent has virtual envy towards the pool and no agent with at least two goods ``sees'' a critical good in the pool,  we will show that the resulting complete allocation returned by the algorithm is $\frac{2}{3}$-EFX. Since at line 22 the allocation is $\alpha$-EFX, for the latter we only need to examine how the agents value their bundles  compared to the bundles that got changed  in line 25  and how the sources that got their bundles changed value their new bundles compared to other bundles.

Initially, when all bundles are empty, Property $\prpt$ holds trivially and the allocation is an orientation. When the algorithm reaches line 22 then by Lemmas \ref{lem:propPreserve_pre}, \ref{lem:propPreserve_post}, \ref{lem:propPreserve_hyper2Middle} and \ref{lem:prop_preserve_fixSources}, inductively, allocation $\X$ is an orientation that satisfies  Property $\prpt$. Moreover, since the while-loop of lines 8-21 terminated with no change in the allocation (breaking condition), we simultaneously have that 
by Observation \ref{obs:pre}, for any agent $i$ it holds,  $\hv_i(X_i)\geq \hv_i(\U(\X))$, 
by Observation \ref{obs:hyper2_upperHalf}  for any agent $i$ with $|X_i|\geq 2$, there are no critical goods in $\U(\X)$, and 
by Observation \ref{obs:post}, for any source $s$ of $\G(\X)$, we have that  $|X_s|=0$ or $|X_s|\geq2$.

Consider any source $s$ that  will get extra 
goods by line 25, and any other agent j $\neq$ $s$ that is not a source. Let $Y_s$ be the set of goods that are assigned to $s$ by the while-loop of lines 24-26.
It holds that $v_{s}(X_{s}) \leq v_{s}(X_{s} \cup Y_s)$  
and thus $s$ will be EFX (and $\alpha$-EFX) if $|X_s|\leq1$, or $\alpha-$EFX,if $|X_s|\geq2$, towards $j$. 
It remains now to examine how an agent $j$ will value the new bundles of the sources, compared to their current bundle $X_j$. 
Let now $j$ be any other agent that is not a source. We will show that  $v_j(X_j)\geq\frac{2}{3}v_j(X_s\cup Y_s)$, which suffices to prove the theorem since even if agent $j$ is  a source herself then for her new bundle it holds,  $v_j(X_j\cup Y_j)\geq v_j(X_j)$. 

If $|X_s|=0$ then $j$ will not envy $s$ no matter how many goods from the pool $Y_s$ contains, since, by Observation \ref{obs:pre}, $\hv_j(X_j)\geq \hv_j(\U(\X))\geq \hv_j(Y_s\cup\{h\})\Rightarrow v_j(X_j)\geq v_j(Y_s\cup\{h\})=v_j(Y_s)$. 
Also, if $Y_s= \emptyset$ then $s$ does not get her bundle changed and, by Property $\prpt$, $j$ will be EFX (and $\alpha$-EFX), if $|X_j|\leq1$, or 
$\alpha-$EFX , if $|X_j|\geq2$, 
towards $s$. Thus, we also focus on sources with $|Y_s|\geq 1$ and let $g_\ell$ be the last good added to the bundle of $s$.

If $|X_{j}| \geq 2$, since $s$ is not virtually envied until it gets $g_\ell$, then 
$\hv_{j}(X_{j}) \geq \hv_{j}(X_{s}\cup Y_s\setminus\{g_\ell\})$ which by $|X_j|\geq2$ and $|X_{s}\cup Y_s\setminus\{g_\ell\}|\geq 2$ implies $v_{j}(X_{j}) \geq v_{j}(X_{s}\cup Y_s\setminus\{g_\ell\})$. 
On the other hand, there is no critical good in the pool for $j$ and thus  $v_{j}(g_\ell) \leq \frac{1}{2}v_{j}(X_{j})$ (Observation \ref{obs:hyper2_upperHalf}). 
Putting these together and using the additivity of the valuation functions we get 
$$v_{j}(X_{s} \cup Y_s) = v_{j}(X_{s}\cup Y_s\setminus \{g_\ell\}) + v_{j}(g_\ell) \leq v_{j}(X_{j}) + \frac{1}{2}v_{j}(X_{j}) = \frac{3}{2}v_{j}(X_{j})$$
\noindent and so j will be $\frac{2}{3}-$EFX towards $s$.

If $|X_{j}| \leq 1$, since $s$ is not virtually envied until it gets $g_\ell$, then 
again $\hv_{j}(X_{j}) \geq \hv_{j}(X_{s}\cup Y_s\setminus\{g_\ell\})$ which now by $|X_j|\leq 1$ and $|X_{s}\cup Y_s\setminus\{g_\ell\}|\geq 2$ implies $v_{j}(X_{j}) \geq \frac{1}{\alpha}v_{j}(X_{s}\cup Y_s\setminus\{g_\ell\})$.
On the other hand,  by  Observation \ref{obs:pre}, $\hv_j(X_j)\geq \hv_j(\U(\X))\geq \hv_j(\{g_\ell ,h\}\})\Rightarrow v_j(X_j)\geq \frac{1}{\alpha}v_j(\{g_\ell\} \cup \{h\})=\frac{1}{\alpha}v_j(g_\ell)$. 
Adding the inequalities and using the additivity of the valuation functions we get 
$$2v_{j}(X_{j}) \geq \frac{1}{\alpha}v_{j}(X_{s}\cup Y_s\setminus\{g_\ell\}) + \frac{1}{\alpha}v_{j}(g_\ell) = \frac{1}{\alpha}v_{j}(X_{s}\cup Y_s)\Rightarrow v_{j}(X_{j}) \geq  \frac{1}{2\alpha}v_{j}(X_{s}\cup Y_s)$$ 
\noindent which by setting $\alpha = \frac{2}{3}$ implies that agent $j$ is $\frac{3}{4}-$EFX towards $s$ and thus $\frac{2}{3}-$EFX towards $s$. 
\end{proof}

\subsection{On the running time of Algorithms \ref{alg:g3m1} and \ref{algo:hyper2}}
\label{sec:RunTimes}

Every loop of  Algorithm \ref{algo:prePro}'s while-loop runs in $poly(n,m)$ time (specifically in $O(nm)$ time). On the other hand since the lexicographic potential is upper bounded and by Observation \ref{obs:pre} whenever the allocation changes after the execution of a loop of the  while-loop the lexicographic potential increases, we have that this while-loop may be executed at most $poly(n,m,2^k)$ times, where here and throughout $k$ is the number of bits needed for representing the values of the valuation functions. Thus, in total Algorithm \ref{algo:prePro} runs in $poly(n,m,2^k)$ time.

Moving on to examine the running time of Algorithm \ref{alg:g3m1} we first argue that every loop of the algorithm's outer while-loop (lines 5-13) needs $poly(n,m,2^k)$ time.
This is because, by above, Algorithm \ref{algo:prePro} runs in $poly(n,m,2^k)$ time, lines 8-10 need $poly(n,m)$ time and line 11 calls Algorithm \ref{algo:postPro} that needs $poly(n,m,2^k)$ time, since, as before, the lexicographic potential is upper bounded and by Observation \ref{obs:post} whenever the allocation changes after the execution of a loop of the  while-loop of Algorithm \ref{algo:postPro} the lexicographic potential increases. On the other hand, by Observation \ref{obs:hyper1}, whenever the allocation changes after the execution of a loop of the outer while-loop of Algorithm \ref{alg:g3m1} the lexicographic potential increases, which implies, since the lexicographic potential is upper bounded, that this while-loop will be executed in at most $poly(n,m,2^k)$ time. Putting it together we get the following lemma.

\begin{lemma}
\label{lem:runTime1}
Algorithm \ref{alg:g3m1}  runs in $poly(n,m,2^k)$ time, where $k$ is the number of bits needed for representing the values of the valuation functions.
\end{lemma}

To examine the running time of Algorithm \ref{algo:hyper2} we act similarly checking that every loop of the outer while-loop of Algorithm \ref{algo:hyper2} runs in $poly(n,m,2^k)$ time and there can be at most $poly(n,m,2^k)$ such loops. To see that every loop of the outer while-loop of Algorithm \ref{algo:hyper2} runs in $poly(n,m,2^k)$ time we observe that (i) similar to the above analysis for Algorithm \ref{alg:g3m1} the while-loop of lines 8-21 need $poly(n,m,2^k)$ time, (ii) the while-loop fo lines 24-26 needs $poly(n,m)$ time, (iii) line 32 calls CycleResolution which needs $poly(n,m)$ time, and (iv) all other commands/routines need $poly(n,m)$ time. On the other hand, by Observation \ref{obs:hyper2_lowerHalf}, whenever the allocation changes after the execution of a loop of the outer while-loop of Algorithm \ref{algo:hyper2} the lexicographic potential increases, which implies that, since the lexicographic potential is upper bounded, this while-loop will be executed in at most $poly(n,m,2^k)$ time. Putting it together we get the following lemma.

\begin{lemma}
\label{lem:runTime2}
Algorithm   \ref{algo:hyper2} runs in $poly(n,m,2^k)$ time, where $k$ is the number of bits needed for representing the values of the valuation functions.
\end{lemma}

\section{Conclusion}

The problem of EFX existence for indivisible goods has received significant attention over the past decade. A very recent breakthrough by \citet{akrami2026counterexample} settles this question negatively for valuation functions beyond the additive class, though the problem remains wide open for additive valuations. Consequently, alternative research directions focus on more restricted settings, even when dealing with general valuations, such as restricted multi-hypergraphs,\footnote{The unrestricted multi-hypergraph setting is equivalent to the general problem.} or on more relaxed fairness notions. This work lies at the intersection of these directions, pushing the frontiers forwards.

\begin{acks}
    The research project is implemented in the framework of H.F.R.I call “Basic research Financing (Horizontal support
    of all Sciences)” under the National Recovery and Resilience Plan “Greece 2.0” funded by the European Union-
    NextGenerationEU (H.F.R.I. Project Number:15635). This work has been partially supported by project MIS 5154714 of the National Recovery and Resilience
    Plan Greece 2.0 funded by the European Union under the NextGenerationEU Program. The work of I.K. was co-funded by the European Union under the project Robotics and advanced industrial production (reg. no. CZ.02.01.01/00/22\_008/0004590) and also supported by the Grant Agency of the Czech Technical University in Prague, grant No.SGS26/181/OHK3/3T/18.
\end{acks}

\bibliographystyle{ACM-Reference-Format}
\bibliography{CombinedRefs}

\newpage
\appendix
\section{Minimal Envied Set Algorithm}

\begin{algorithm}[h]
\caption{\MinEnviedSet($\X, j$) (Algorithm 2.3 of \cite{CKMS21} )}
\label{alg:minEnviedSet}
\raggedright\textbf{Input:} A partial allocation $\X$ that at least one agent $j$ envies $\U(\X)$, according to virtual valuations.

\raggedright\textbf{Output:} A (minimal) subset $T$ of $\U(\X)$, and an agent $j$, such that $\hv_j(T)>\hv_j(X_j)$, and for any agent $i$ and $S \subset T$, $\hv_i(X_i)\geq \hv_i(S)$. 
\begin{algorithmic}[1]
\STATE $T=\U(\X)$
\FOR {every agent $i$}
\FOR {every good $g\in T$}
\IF{$\hv_i(T\setminus\{g\})>\hv_i(X_i)$}
\STATE $T\gets T\setminus\{g\}$
\STATE $j = i$
\ENDIF
\ENDFOR
\ENDFOR
\RETURN $(T,j)$
\end{algorithmic}
\end{algorithm}

\end{document}